\documentclass[journal]{IEEEtran}

\IEEEoverridecommandlockouts                        

\usepackage{amsmath}
\usepackage{amsfonts}
\usepackage{graphicx}
\usepackage{algorithm}
\usepackage{subcaption}
\usepackage[noend]{algpseudocode}
\usepackage{cite}
\usepackage{bigints}
\usepackage[short]{optidef}
\usepackage{amsthm}
\usepackage{xcolor}

\newcommand{\tr}{\mathrm{tr}}
\newcommand{\dt}{\,{dt}}
\newcommand{\dq}{\,{dq}}
\newcommand{\dtprime}{\,{dt'}}

\newcommand{\norm}[1]{\left\lVert#1\right\rVert}

\newtheorem{proposition}{Proposition}
\newtheorem{lemma}{Lemma}

\newtheorem{assumption}{Assumption}

\def\Vec{\mathop{vec}\nolimits}

\title{\LARGE \bf
Multi-Agent Persistent Monitoring of Targets with Uncertain States
}

\author{Samuel C. Pinto, {\it Student Member, IEEE}, Sean B. Andersson, {\it Senior Member, IEEE}, Julien M. Hendrickx, {\it Member, IEEE} and Christos G. Cassandras, {\it Fellow, IEEE}
\thanks{The authors are with the Dept. of Mechanical Engineering (Pinto and Andersson), the Division of Systems Engineering (Andersson, Cassandras) and the Dept. of Electrical and Computer Engineering (Cassandras), Boston University, Boston, MA, 02215, USA. (e-mail: \{samcerq,sanderss,cgc\}@bu.edu)

Julien M. Hendrickx is with the ICTEAM Institute, UCLoucain, Louvain-la-Neuve 1348, Belgium. (e-mail: julien.hendrickx@uclouvain.be)

This work was supported in part by NSF under grants ECCS-1931600, DMS-1664644, CNS-1645681, and CMMI-1562031, by ARPA-E's NEXTCAR program under grant DE-AR0000796, by AFOSR under grant FA9550-19-1-0158,  and by the MathWorks. The work of J. Hendrickx was supported by the “RevealFlight” Concerted Research Action (ARC) of the Federation Wallonie-Bruxelles, by the Incentive Grant for Scientific Research (MIS) “Learning from Pairwise Data” of the F.R.S.-FNRS, and by a WBI World Excellence Fellowship.
}
}

\begin{document}

\maketitle
\thispagestyle{empty}
\pagestyle{empty}

\begin{abstract}                
We address the problem of persistent monitoring, where a finite set of mobile agents has to persistently visit a finite set of targets. Each of these targets has an internal state that evolves with linear stochastic dynamics. The agents can observe these states, and the observation quality is a function of the distance between the agent and a given target. The goal is then to minimize the mean squared estimation error of these target states. We approach the problem from an infinite horizon perspective, where we prove that, under some natural assumptions, the covariance matrix of each target converges to a limit cycle. The goal, therefore, becomes to minimize the steady state uncertainty. Assuming that the trajectory is parameterized, we provide tools for computing the steady state cost gradient. We show that, in one-dimensional (1D) environments with bounded control and non-overlapping targets, when an optimal control exists it can be represented using a finite number of parameters. We also propose an efficient parameterization of the agent trajectories for multidimensional settings using Fourier curves. Simulation results show the efficacy of the proposed technique in 1D, 2D and 3D scenarios.
\end{abstract}


\section{Introduction}

We consider the problem of multi-agent persistent monitoring. This problem consists of using a finite set of agents to monitor a finite set of targets, more numerous than agents, which have internal states that evolve over time with dynamics subject to uncertainty. Therefore, as time goes to infinity, in order to keep the uncertainty under control, the targets need to be visited only a finite number of times, but persistently. The goal is to minimize the long-term uncertainty by designing movement policies that produce the best estimate possible of the target states. This paradigm finds applications across a wide range of domains, such as trajectory planning of underwater vehicles to measure ocean temperature \cite{lan2013planning,lan2014variational,Alam:2018ie}, surveillance in smart cities \cite{kim2018designing} and tracking of multiple microparticles by an optical microscope \cite{shen2010tracking}. 

This problem is closely related to the Multi Traveling Salesman Problem (MTSP) \cite{BEKTAS2006209} and Multi-Vehicle Routing Problem (MVRP) \cite{laporte2009fifty}, where, given a set of targets (possibly constrained to a graph-based structure), the goal is to find a cycle in which the agents efficiently visit all the targets in order to minimize the traveled distance or total travel time. These problems are proved to be computationally intractable (NP-hard) and most of the scalable solutions to these problems rely either on local optimization or heuristics \cite{pasqualetti2012cooperative,BEKTAS2006209,laporte2009fifty}. The major difference between the MTSP and MVRP and the problem we are dealing with in this paper is that the optimization goal we consider is to minimize the uncertainty rather than distance or time between two consecutive observations of a given target. The present work is also closely related to the sensor allocation problem \cite{le2010scheduling}, where a set of sensors can observe a set of targets, but due to various constraints not all the targets can be observed at the same time and, therefore, some of the sensors have to switch among the targets they observe. The sensor allocation problem, however, assumes that the sensors are fixed and therefore does not incorporate the effect of the agent movement (i.e. the mobile sensors) in the formulation.

In the realm of persistent monitoring, significant previous work has been done. In \cite{lan2013planning}, a variant of the Rapid-Exploring Random Tree (RRT) algorithm was designed for cyclic Persistent Monitoring and \cite{lan2014variational} introduced an optimal control approach that relied on a solution of the two-point boundary value problem resulting from a Hamiltonian analysis. Note that the solution of the two-point value problem is numerically challenging and computationally expensive. In \cite{jones2015information} the persistent monitoring problem is formulated using temporal logic to encode target visiting constraints rather than solving an optimization problem.

The present paper also builds up from significant previous work by the authors \cite{cassandras2013optimal,zhou2018optimal,Yu:2018cf}, where the problem of persistent monitoring was modelled using an uncertainty metric for each of the targets that either grew linearly with time when the agent was not observed or decreased linearly when an agent visited it. A common feature among these previous works and the present paper is the focus on scalable solutions with respect to the number of agents, targets and time horizon. Therefore, instead of looking for globally optimal visiting schedules, we use a local optimization scheme (gradient descent) even though the obtained solution is not guaranteed to be globally optimal. One big challenge in order to use a gradient descent approach is to efficiently compute the gradients of the cost with respect to the parameters that define the trajectory.

The current work, unlike some previous work by the authors, considers each target as having an internal state that evolves with linear stochastic dynamics that can be observed with a linear observation model. The signal to noise ratio of the observation is a function of the distance between the agent and the target. In this setting, the optimal estimator can be proven to be a Kalman-Bucy filter and the mean estimation error is directly related to the covariance matrix of this filter. The main contribution of this paper is to provide tools to efficiently represent and optimize the schedules for agents visiting targets. If we consider finite horizon schedules, as time grows to infinity, the number of parameters to represent a trajectory also tends to grow infinitely large. We, however, restrict ourselves to a periodic trajectory and approach the problem from an infinite horizon perspective. We show that under some very natural assumptions the estimation error converges to a limit cycle and we provide tools for optimizing one period of the limit cycle trajectory, which usually is represented by only a small number of parameters.

Although the analysis introduced in this paper is independent of the particular parameterization chosen for the trajectory, we discuss two parameterizations that are particularly interesting. When the targets and agents are constrained to lie in a one-dimensional environment, we show that, under some assumptions, an optimal control can always be represented by a trajectory in which the agent is either moving with full speed or dwelling at a fixed position. This allows optimal trajectories to be described as a finite sequence of movement times and dwelling times, yielding a parameterization. On the other hand, when the agents and targets operate in a higher dimensional space (e.g. 2D and 3D), we cannot immediately extend such properties of an optimal control. We then parameterize the trajectories using Fourier curves, where the movement of an agent in each of the coordinates is described by a truncated Fourier series. Fourier curves are interesting because they are able to describe very general smooth movement policies with only a very small number of coefficients.

Recalling the goal of performing local optimization using a gradient descent scheme, it is particularly important to provide good initial solutions for the optimization. We thus connect the persistent monitoring problem with the MTSP and use a heuristic solution to the MTSP as a basis for the initial trajectory of the agents in the optimization scheme. We benefit from the fact that efficient heuristic solutions of the MTSP are well studied in the scientific literature and that they always provide an initial trajectory where all the targets are visited. This is a very important feature for persistent monitoring, since it prevents the uncertainty of each target from becoming infinitely large.


Preliminary results of this work have appeared in previous publications. In \cite{pinto2019monitoring}, the target internal state dynamics and observation models, as well as the 1D transient analysis were introduced. The computation of steady state gradients and infinite horizon analysis was first introduced in \cite{pinto2020monitoring}. The extension to multi dimensional environments using Fourier curves was initially presented in \cite{pinto2020multidimensional}. However, the approach described in our previous works was heavily dependent on the specific parameterization and in the present work we formulate the problem in a general framework that does not rely on the specific parameterization. Moreover, in Sec. \ref{sec:transient_opt} we provide a proof that guarantees the convergence and uniqueness of the steady state covariance matrices and also we show the soundness of our method to compute the steady state gradients. On top of that, we provide a stronger claim than we did in \cite{pinto2019monitoring} about an optimal parameterization of 1D trajectories. {Previously, we were only able to show that it can be parameterized, but now we provide an explicit bound on the number of parameters}. Moreover, in Sec. \ref{sec:fourier_curves} we include simulation results that significantly add to the results of our previous work \cite{pinto2020multidimensional}. 

The rest of this paper is organized as follows. Section \ref{sec:prob_formulation} describes the  models used for the agents and the target internal states, along with a formulation of the optimal joint control and estimation problem. Section \ref{sec:transient_opt} presents results on the convergence of the covariance matrix and the optimization procedure is given for the periodic, infinite horizon case. Section \ref{sec:parameterizations} introduces the 1D parameterization, along with its properties, optimization initialization and some simulation results. In Section \ref{sec:fourier_curves}, some features of the previous section are extended to higher dimensions using Fourier curves and 1D, 2D and 3D results are presented.  Finally Section \ref{sec:conclusion} gives a conclusion and shares ideas for future works.

\section{Problem Formulation}
\label{sec:prob_formulation}

Consider an environment with a set of $M$ points of interest (targets) at fixed positions $x_i\in \mathbb{R}^P$, $i=1,...,M$. Each of these targets has an internal state $\phi_i \in \mathbb{R}^{L_i}$ that needs to be monitored and that evolves according to linear time-invariant stochastic dynamics:
\begin{equation}
    \label{eq:dynamics_phi}
    {\dot{\phi}}_i(t) = A_i{\phi}_i(t) + {w}_i(t),
\end{equation}
where $w_i(t)$ is a white noise process distributed according to $w_i(t) \sim \mathcal{N}(0,Q_i),$ $i=1,\dots, M,$ and $w_i(t)$ and $w_j(t)$ are statistically independent if $i \neq j$.

Suppose that there is a collection of $N$ mobile agents at positions $s_i(t)\in \mathbb{R}^P$ that can move with the following kinematic model:
\begin{equation}
    \dot{s}_j(t)=u_j(t),\qquad u_j(t)\in\mathcal{U},\qquad j=1,...,N, \label{eq:dynamics_agents}
\end{equation}
where $u_j$ is an input, and $\mathcal{U}$ is the set of admissible inputs. Even though we assume, for the sake of simplicity, first order dynamics and that only the speed may be bounded, the results in this paper could be extended to more complex dynamics and constraints. For example, \cite{wang2018etal} explored similar results in a simplified version of the persistent monitoring problem, considering double integrator agent dynamics with constraints both on the speed and the acceleration.

Each of these agents is equipped with sensors that can observe the targets according to the following model:
\begin{equation}
    \label{eq:observation_model_ij}
    {z}_{i,j}(t)=\gamma_j\left(s_j(t)-x_i\right)H_i{\phi}_i(t)+{v}_{i,j}(t),
\end{equation}
where $v_{i,j}(t)$ is a white noise process distributed according to $v_{i,j}(t) \sim \mathcal{N}(0,R_i)$ with $v_{i,j}(t)$ independent of $v_{k,l}$ if $i\neq k$ or $j\neq l$, and $\gamma_{i,j}: \mathbb{R}^N\mapsto \mathbb{R}$ is a function that captures the interdependence of measurement quality and the relative position from a given agent to a target. The intuition behind this function is that the instantaneous signal to noise ratio (SNR) can be computed as:
\begin{multline}
    \label{eq:snr}
    \frac{E\left[\norm{{z}_{i,j}(t)-{v}_{i,j}(t)}^2\right]}{E[\norm{{v}_{i,j}(t)}^2]} =\gamma_{i,j}^2\left(s_j(t)-x_i\right)\frac{\norm{H_i{\phi}_i(t)}^2}{\text{tr}(R_i)},
\end{multline}
where $\tr (\cdot)$ is the trace of the matrix. Notice that the term $\norm{H_i{\phi}_i(t)}^2(\text{tr}(R_i))^{-1}$ is a deterministic scalar that does not depend on the relative position between the target and the agent. Therefore, the function $\gamma_{i,j}$ captures entirely how the position of the agent affects the quality of the measurement. It is worth noting that in most of the applications of mobile agents to sensing there is a limited sensing range or the quality of the measurement gets worse as the agent moves farther away from the target. The general model of $\gamma_{i,j}$ is capable of capturing both the finite range and the dependence between measurement quality and relative position of the target from the agent. Even though the analysis in this paper does not depend on the specific $\gamma_{i,j}$, for the sake of concreteness we use the following form:
\begin{equation}
    \label{eq:model_gamma}
    \gamma_{i,j}(\alpha) =    
    \begin{cases}
        0, & \norm{\alpha}>r_{i,j},\\
        \sqrt{1-\frac{\norm{\alpha}}{r_{i,j}}}, & \norm{\alpha}\leq r_{i,j}.
    \end{cases}
\end{equation}

The intuition behind this specific form is that the best measurement quality is achieved when the agent's location coincides with that of the target, with the SNR decaying linearly as the agent moves away. When agent is at a distance larger or equal to its sensing radius $r_{i,j}$, only noise is observed.

In this paper we approach the problem from a centralized perspective. Therefore, at a given instant, the combined observations from all the agents of a single target can be grouped in a vector $\tilde{z}(t)$ as:
\begin{equation}
    \label{eq:vector_observation_model}
    {z}_i(t) = \begin{bmatrix}{z}_{i,1}' & ... & {z}_{i,N}' \end{bmatrix}' = \tilde{H}_i(s_1,...,s_n){\phi}_i(t)+\tilde{{v}}_i(t),
\end{equation}
where
\begin{align}
    \label{eq:def_h_tilde}
    \tilde{H}_i&=\begin{bmatrix}\gamma_1(s_1-x_i)H_i' & \cdots & \gamma_N(s_N-x_i)H_i'\end{bmatrix}', \\
    \tilde{{v}}_i(t)&=\begin{bmatrix} {v}_{i,1}'(t) & ... & {v}_{i,N}'(t)\end{bmatrix}',
    \label{eq:def_r_tilde}
 \end{align}
 and, since $v_{i,j}(t)$ is independent of $v_{i,k}(t)$ if $k\neq j$,
 \begin{align}
    E[\tilde{ v}_i'(t)\tilde{{ v}}_i(t)]&=\tilde{R}_i=
    \begin{bmatrix} 
    {R}_{i} & {0} & \dots & {0}\\
    {0} & R_i & \dots & {0}\\
    \vdots & \vdots  &\ddots & \vdots \\
    {0}& {0} & \dots   &     R_i 
    \end{bmatrix}.
\end{align}

 The overall goal is to obtain estimators $\hat{\phi}_i(t,z(t))$ and open-loop control inputs $u_j(t)$ to minimize the following cost function:
\begin{equation}
    \label{eq:cost_function_expectation}
    J=\frac{1}{t_f}\bigintssss_{0}^{t_f}\left( \sum_{i=1}^ME[e_i'(\zeta)e_i(\zeta)]+\beta \sum_{j=1}^Nu_j'(\zeta)u_j(\zeta)\right)d \zeta,
\end{equation}
where ${e_i}(t)=\hat{{\phi}}_i(t)-{{\phi}}_i(t)$ and $t_f$ is the time horizon. This cost function represents a weighted sum of the mean squared estimation error and the control effort; thus, the weighting factor $\beta$ is responsible for balancing the importance of these two optimization goals. 

The models in \eqref{eq:def_h_tilde} and \eqref{eq:def_r_tilde} define a linear time-varying stochastic system. Based on a similar statement from \cite{lan2014variational}, we have the following proposition:

\begin{proposition}
    The optimal unbiased estimator $\hat{\phi_i}$ for the the cost function \eqref{eq:cost_function_expectation}, dynamics \eqref{eq:dynamics_phi}, and observation model \eqref{eq:observation_model_ij}, is the Kalman-Bucy filter, given by:
    \begin{subequations}
	\begin{align}
    \dot{\hat{{\phi}}}_i(t)&=A_i\hat{{\phi}}_i(t)+\Omega(t)_i\tilde{H}_i'(t)\tilde{R}_i^{-1}\left(\tilde{{z}}_i(t)-\tilde{H}_i(t)\hat{{\phi}}_i(t)\right), \\
    \dot{\Omega}_i(t) &= A_i\Omega_i(t)+\Omega_i(t)A_i'+Q_i-\Omega_i(t)\tilde{H}_i'\tilde{R}^{-1}_i\tilde{H}_i\Omega_i(t), \label{eq:dynamics_omega_matricial}   
	\end{align}
\end{subequations}
where $\Omega_i(t)$ is the covariance matrix of the estimator.
\end{proposition}
\begin{proof}
    See Appendix \ref{ap:proof_kbfilter}.
\end{proof}

Using \eqref{eq:def_h_tilde} and \eqref{eq:def_r_tilde}, we can rewrite \eqref{eq:dynamics_omega_matricial} as:
\begin{equation}
    \label{eq:dynamics_omega}
    \dot{\Omega}_i(t) = A_i\Omega_i(t)+\Omega_i(t)A_i'+Q_i
    -\Omega_i(t)G_i\Omega_i(t)\sum_{j=1}^N\gamma_{i,j}^2(t),
\end{equation}
where $G_i=H_i'R_i^{-1}H_i$ and $\gamma_{i,j}(t)=\gamma_{i,j}(s_j(t)-x_i)$. Using the fact that 
\begin{align*}
	E\left[{e}_i'(t){e}_i(t)\right]=\text{tr}(E\left[{e}_i(t){e}'_i(t)\right])=\tr(\Omega_i(t)),	
\end{align*}
we can rewrite the cost function in \eqref{eq:cost_function_expectation} as
\begin{equation}
    \label{eq:cost_function}
    J=\frac{1}{t_f}\bigintsss_{0}^{t_f}\left( \sum_{i=1}^M\tr(\Omega_i(\zeta))+\beta \sum_{j=1}^Nu_j'(\zeta)u_j(\zeta)\right)d \zeta.
\end{equation}
The goal is then to minimize the cost \eqref{eq:cost_function} subject to the dynamics in \eqref{eq:dynamics_omega} and \eqref{eq:dynamics_agents}. {In other words, we aim to design a trajectory, with constrained controls $u_j\in\mathcal{U}$, and estimation error linked to the trajectory through the dynamics of the covariance matrix of the Kalman-Bucy Filter that minimizes a weighted sum of the total control effort and the mean estimation error.} 

\section{Optimization of Parameterized Trajectories}
\label{sec:transient_opt}
\subsection{Finite Horizon Trajectory Optimization}
\label{subsec:transient_opt}
Even though we focus on the optimization of infinite horizon trajectories, we briefly review the procedure for optimizing trajectories with a finite time horizon in order to later extend to the infinite horizon setting. 
In this section, we establish a general formulation, where we assume that the agent trajectories can be fully defined by a finite set of parameters. In the following sections we approach specific settings that show that parameterizations tend to naturally fit the persistent monitoring problem. Our overall goal is to compute locally optimal solutions with respect to these parameters using gradient descent. Therefore, we initially discuss how to compute the gradients for the finite horizon version of the problem. We define the set of parameters that fully describe the trajectory for $t\in[0,t_f]$ as $\Theta = \{\theta_1,...,\theta_D\}$.


Recalling the expression for the cost \eqref{eq:cost_function}, we can compute the partial derivative with respect to one of the parameters of the trajectory $\theta_d$ as:
\begin{equation}
    \label{eq:transient_gradient}
    \frac{\partial J}{\partial \theta_d}=\frac{1}{t_f}\bigintsss_{0}^{t_f}\left( \sum_{i=1}^M\tr\left(\frac{\partial \Omega_i}{\partial \theta_d}(\zeta)\right)+\beta \sum_{j=1}^N\frac{\partial (u_j'u_j)}{\partial \theta_d}(\zeta)\right)d \zeta.
\end{equation}

Note further that, given the dynamics of the covariance matrix in \eqref{eq:dynamics_omega}, $\frac{\partial \Omega_i}{\partial \theta_d}$ is the solution of the following ODE:
\begin{multline}
    \label{eq:derivative_dynamics_omega}
    \frac{d}{dt}\left(\frac{\partial {\Omega}_i(t)}{\partial \theta_d}\right) = A_i\frac{\partial \Omega_i}{\partial\theta_d}(t)+\frac{\partial \Omega_i(t)}{\partial \theta_d}A_i'+Q_i
    \\-\left(\frac{\partial{\Omega}_i(t)}{\partial \theta_d}G_i\Omega_i(t)+\Omega_i(t)G_i\frac{\partial \Omega_i(t)}{\partial \theta_d}\right)\sum_{j=1}^N\gamma_{i,j}^2(t)\\
    -\Omega_i(t)G_i\Omega_i(t)\sum_{j=1}^N\frac{\partial \gamma_{i,j}^2}{\partial \theta_d}(t),
\end{multline}
with initial conditions $\frac{\partial \Omega_i}{\partial \theta_d}(0)=0$. Also, we know that
\begin{equation}
    \frac{\partial \gamma_{i,j}^2(t)}{\partial \theta} = \sum_{p=1}^P\frac{\partial \gamma_{i,j}^2(t)}{\partial s_j^{e_p}}\frac{\partial s_j^{e_p}(t)}{\partial \theta_d},
\end{equation}
where $e_p$, $p=1,...,P$ is the $p$-th coordinate of the space where the agents move in. Given the specific definition of $\gamma_{i,j}$ in \eqref{eq:model_gamma}, we can easily see that
\begin{equation}
    \frac{\partial \gamma_{i,j}^2}{\partial s_j^{e_i}} = \begin{cases}
    \frac{s_j^{e_p}-x_i^{e_p}}{r_j\norm{s_j-x_i}},\ &\ \text{if } \norm{s_j-x_i}<r_j,\\
    0,\ &\text{ otherwise}.
    \end{cases}
\end{equation}

The only terms that we have not yet given a procedure to compute are $\frac{\partial (u_j'u_j)}{\partial \theta_d}(t)$ and $\frac{\partial s_j^{e_p}}{\partial \theta_d}(t)$. The computation of both of these terms is intrinsically related to the specific parameterization chosen and details of their computation will be discussed in Secs. \ref{sec:parameterizations} and \ref{sec:fourier_curves}. Note that we use the partial derivatives of the covariance matrices in \eqref{eq:transient_gradient} in order to compute the gradient of the cost $J$. The complete procedure to compute the transient problem gradients is given in Alg. \ref{alg:transient_optimization}.

\begin{algorithm}
\caption{Transient Gradient Computation}
\label{alg:transient_optimization}
\begin{algorithmic}[1]
\Procedure {ComputeTransientGradient}{}
\State{\bf Input}: $\Theta$
\State Compute $s_1(t),...,s_N(t)$ from the parameterization
\For{every $\theta$ in $\Theta$}
\State{Compute $\frac{\partial}{\partial \theta}\int_0^{t_f}\sum_{j=1}^N{u}_j'(\zeta){u}_j(\zeta)d\zeta$ according to the parameterization}
\State{Compute $\frac{\partial s_j(t)}{\partial \theta}$ according to the parameterization}
\For {$i$ ranging from $1$ to $M$}
\State{Compute $\frac{\partial \Omega_i(t)}{\partial \theta}$ by solving ODE \eqref{eq:derivative_dynamics_omega}}
\EndFor

\State Compute $\frac{\partial J}{\partial \theta}$ using \eqref{eq:transient_gradient}.
\EndFor
\State {\bf Output:} $\nabla J$
\EndProcedure
\end{algorithmic}
\end{algorithm}

\subsection{Steady State Persistent Monitoring}
\label{sec:periodic_persistent_monitoring}

For a persistent monitoring task to be successful, it is necessary that targets are visited infinitely often as time goes to infinity, because otherwise their uncertainty can become unbounded. Periodicity naturally fits into the persistent monitoring paradigm, since targets need to be visited infinitely often and, although a periodic structure of the solution is not necessarily optimal, simulation results in the transient case show that the trajectories tend to converge to oscillatory behavior \cite{pinto2019monitoring}. On top of that, periodicity provides an upper bound to the inter-visit time. 
Moreover, if periodicity is assumed, the infinite horizon trajectory is fully defined by the trajectory of a single period. This often leads to needing only a very small number of parameters to describe the infinite horizon trajectory and, as a consequence, only a small number of parameters have to be optimized in order to generate efficient trajectories. With that in mind, in this section we explore the properties of periodic solutions to the persistent monitoring problem when the system fulfills the following very natural assumptions.

\begin{assumption}
    The pair $(A_i,H_i)$ is detectable, for every $i\in\{1,...,M\}$.
\end{assumption}
\begin{assumption}
    $Q_i$ and the initial covariance matrix $\Sigma_i(0)$ are positive definite, for every $i\in\{1,...,M\}$.
\end{assumption}

The intuition behind the first assumption is that it ensures that sensing can guarantee that the uncertainty of each target will be bounded even for long horizons. The second one ensures that the covariance matrix will always be positive definite, a fact that will be used to prove Prop. \ref{prop:solution_lyapunov_eq}. The results in this paper would likely still hold if Assumption 2 was relaxed, even though the proof of Prop. \ref{prop:solution_lyapunov_eq} could become more complex. Under these assumptions, first we explore conditions under which the convergence of the covariance matrix is achieved. For the sake of notational conciseness, we define
\begin{equation}
    \label{eq:definition_eta}
    \eta_i(t)=\sum_{j=1}^N\gamma_{i,j}^2(t),
\end{equation}
which represents the instantaneous power level of the sensed signal, combining all the agents' observations of the same target $i$. Using a procedure similar to the one used in the proof of Lemma 9 in \cite{le2010scheduling}, we establish the following proposition:

\begin{proposition}
    \label{prop:unique_attractive_sol_riccati_eq}
    If $\eta_i(t)$ is $T$-periodic and $\eta_i(t) > 0$ for some non-degenerate interval $[a,b]\in[0,T]$, then, under Assumption 1, there exists a unique non-negative stabilizing $T$-periodic solution to \eqref{eq:dynamics_omega}.
\end{proposition}
\begin{proof}
    According to \cite[p.~130]{bittanti1984periodic}, a pair $(A_i,\eta_i(t)H_i)$ of a periodic system is detectable if and only if for every eigenpair $(x,\lambda)$ with $x\neq0$,
    \begin{multline}
        A_ix=\lambda x \implies \exists\ [a,b]\in[0,T]\ s.t.\ \eta_i(t)e^{\lambda t}H_ix\neq0,
    \end{multline}
    $\forall t\in[a,b]$ and $[a,b]$ is non-degenerate. Notice that, due to Assumption 1, for any eigenvector $x$ of $A_i$, $H_ix\neq0$. Therefore, when $\eta_i(t)>0$ (i.e. any $t\in[a,b]$), $\eta_i(t)e^{\lambda t}H_ix\neq0$, which implies that $(A_i,\eta_i(t)H_i)$ is detectable. Therefore, the collorary to Theorem 3 in \cite[p.~95]{nicolao1992convergence} shows that there exists a non negative $T$-periodic solution to \eqref{eq:dynamics_omega}, $\bar{\Omega}_i(t)$, and 
    \begin{equation*}
        \lim_{t\rightarrow \infty}(\Omega_i(t)-\bar{\Omega}_i(t))=0
    \end{equation*}
    for any solution $\Omega_i(t)$ with positive definite initial condition $\Omega_i(0)$.
\end{proof}

Prop. \ref{prop:unique_attractive_sol_riccati_eq} implies that, if $\eta_i(t)$ is periodic, given any initial covariance matrix $\Omega_i(0)$, the estimation covariance for target $i$ converges to a $T$-periodic matrix $\bar{\Omega}_i(t)$, as long as target $i$ is visited for some non-zero amount of time in the periodic trajectory. Therefore, 
\begin{equation*}
    \forall \delta>0, \exists\ t_0\ s.t.\ |{\tilde{\Omega}_i(t)-{\Omega}_i(t)}|\leq\delta,\ \forall t\geq t_0,
\end{equation*}
which implies that

\begin{equation}
    \label{eq:infinite_time_steady_state_periodic}
    \lim_{t\rightarrow \infty}\frac{1}{t}\int_{0}^t |\tr(\tilde{\Omega}_i(t')-{\Omega}_i(t'))|\dtprime\leq\delta.
\end{equation}

This discussion implies that, if we run a periodic trajectory for long enough, the mean estimation error will become arbitrarily close to the mean steady state estimation error. Therefore, if we plan only (one period of) the steady state trajectory, the actual estimation error will be arbitrarily close to that of the planned trajectory as time goes to infinity. Even though Prop. \ref{prop:unique_attractive_sol_riccati_eq} states that the solution of the periodic Riccati equation is globally attractive, it does not provide any convergence rate for its numerical computation. However, the problem of computing numerical solutions to this equation has been studied in other works and we refer the reader to \cite{varga2013computational} for a good review and discussion of these methods.

Similarly as in the transient case, we intend to optimize the trajectory of the agents using gradient descent. However, the computation of the steady state gradients of the covariance matrix is more challenging than the transient case discussed in Subsec. \ref{subsec:transient_opt}. In the sequel, we provide the procedure to compute these gradients when they exist.

\subsection{Steady State Gradients}
\label{sec:steady_state}
Assuming that the trajectory is periodic and all the targets are visited, we introduce the change of variable $q=t/T$, where $T$ is the period of the trajectory. The steady state cost can be rewritten as:
\begin{equation}
    \label{eq:cost_function_T}
    J=\bigintsss_{0}^1\left( \sum_{i=1}^M\tr(\bar{\Omega}_i(q))+\beta \sum_{j=1}^N\bar{u}_j'(q)\bar{u}_j(q)\right)\dq,
\end{equation}
where $\bar{u}(q)=u(qT)$. Similar to \eqref{eq:transient_gradient}, we know that, given some parameter $\theta_d\in\Theta$:
\begin{equation}
    \label{eq:derivative_steady_state_cost}
    \frac{\partial J}{\partial \theta_d}=\bigintsss_{0}^1\left( \sum_{i=1}^M\tr\left(\frac{\partial \bar{\Omega}_i(q)}{\partial \theta}\right)+\beta \sum_{j=1}^N\frac{\partial (\bar{u}_j'\bar{u}_j)(q)}{\partial \theta_d}\right)\dq.
\end{equation}
{$\bar{\Omega}_i(q)$, when it exists, is defined by the following dynamics
\begin{multline}
\label{eq:scaled_riccati_diff_eq}
\dot{\bar{\Omega}}_i(q)=\frac{d{\bar{\Omega}}_i(q)}{dq} = T(A\bar{\Omega}_i(q)+\bar{\Omega}_i(q)A'+Q\\-\eta_i(q)\bar{\Omega}_i(q)G\bar{\Omega}_i(q)),
\end{multline}
along with the periodicity condition $\bar{\Omega}_i(0)=\bar{\Omega}_i(1)$. Now, suppose that the gradient of $\bar{\Omega}_i(q)$ with respect to a parameter $\theta_d$ exists. Then, this gradient is the solution of the following differential equation (note that the period may be a function of the parameters or a parameter itself):
\begin{multline}
    \label{eq:derivative_omega_i}
    \dot{\Sigma}(q)-T\biggl(A\Sigma(q)+\Sigma(q)A'-\eta_i(q)\bar{\Omega}_i(q)G\Sigma(q)\\-\eta_i(q)\Sigma(q)G_i\bar{\Omega}_i(q)\biggr) = T\frac{\partial \eta_i(q)}{\partial \theta_d}\bar{\Omega}_i(q)G_i\bar{\Omega}_i(q)+\frac{\partial T}{\partial \theta_d}\frac{\dot{\bar{\Omega}}_i}{T},
\end{multline}
with periodicity conditions $\Sigma(0)=\Sigma(1)$. In order to study the computation of $\Sigma(q)$, we define the following auxiliary problems:
\begin{equation}
    \label{eq:homogeneous_derivative_lyapunov_eq}
    \dot{\Sigma}_H-T\left(A-\eta_i\bar{\Omega}_iG\right)\Sigma_H=0,\ \Sigma_H(0)=I,
\end{equation}
\begin{multline}
    \label{eq:zero_initial_derivative_lyapunov_eq}
    \dot{\Sigma}_{ZI}-T\left(A-\eta_i\bar{\Omega}_iG\right)\Sigma_{ZI}-T\Sigma_{ZI}'\left(A-\eta_i\bar{\Omega}_iG\right)'\\
    =T\frac{\partial \eta_i}{\partial \theta}\bar{\Omega}_iG\bar{\Omega}_i+\frac{\partial T}{\partial \theta_d}\frac{\dot\Omega_i}{T},\ \Sigma_{ZI}(0)=0,
\end{multline}
where the time dependence of $\eta_i(q),\ \Omega_i(q), \Sigma_{ZI}(q)$ and $\Sigma_H(q)$ was omited for conciseness. Then, in the following Proposition we exploit these auxiliary problems for computing $\Sigma(q)$.}
\begin{proposition}
    \label{prop:solution_lyapunov_eq}
   Suppose $\Sigma_H$ is a solution of \eqref{eq:homogeneous_derivative_lyapunov_eq}, $\Sigma_{ZI}$ is a solution of \eqref{eq:zero_initial_derivative_lyapunov_eq}, Assumptions $1$ and $2$ hold, and that target $i$ is observed at least once in the period $T$. 
   {Then, the equation
    \begin{equation}
        \label{eq:discrete_lyapunov_equation}
        \Lambda=\Sigma_H(1)\Lambda\Sigma_H'(1)+\Sigma_{ZI}(1)
    \end{equation}
    has a unique solution $\Lambda$. Additionally, when $\frac{\partial {\bar\Omega}_i(q)}{\partial \theta_d}$ exists,
    \begin{equation}
        \frac{\partial {\bar\Omega}_i(q)}{\partial \theta_d} = \Sigma(q)= \Sigma_H'(q)\Lambda\Sigma_H(q)+\Sigma_{ZI}(q).
    \end{equation}}
\end{proposition}

\begin{proof}
Suppose $\Lambda$ and $\tilde{\Lambda}$ are solutions of \eqref{eq:homogeneous_derivative_lyapunov_eq}, then
\begin{equation}
    \Lambda-\tilde{\Lambda} = \Sigma_H(1)\left(\Lambda-\tilde{\Lambda}\right)\Sigma_H'(1) 
\end{equation}
which is equivalent to
\begin{equation}
    \label{eq:vectorization_difference_solutions_Lyapunov}
    \Vec{\left(\Lambda-\tilde{\Lambda}\right)} = \left(\Sigma_H(1)\otimes \Sigma_H(1)\right) \Vec{\left(\Lambda-\tilde{\Lambda}\right)},
\end{equation}
where $\Vec(\cdot)$ is the operator the performs the matrix vectorization and $\otimes$ represents the matrix Kronecker product.
Notice that $\Lambda=\tilde{\Lambda}$ is a solution of \eqref{eq:vectorization_difference_solutions_Lyapunov}. This solution is the unique solution if and only if $1$ is not an eigenvalue of $\Sigma_H(1)\otimes \Sigma_H(1)$. On the other hand, the eigenvalues of $\Sigma_H(1)\otimes \Sigma_H(1)$ are all in the form $\mu_1\mu_2$, where $\mu_1$ and $\mu_2$ are distinct eigenvalues of $\Sigma_H(1)$ \cite{zhang2011matrix}.

In the following we show that all the eigenvalues of $\Sigma_H(1)$ have absolute value lower than one. For that, first notice that since Q is positive definite,  $\bar{\Omega}_i$ is also positive definite and hence, invertible. Define
\begin{equation*}
    \mathcal{W} = \bar{\Omega}_i^{-1},
\end{equation*}
and, since $\dot{\mathcal{W}}=-\bar{\Omega}_i^{-1}\dot{\bar{\Omega}}_i\bar{\Omega}^{-1}_i=-\mathcal{W}\dot{\bar{\Omega}}_i\mathcal{W}$, using \eqref{eq:dynamics_omega} and \eqref{eq:definition_eta}, the dynamics of $\mathcal{W}$ can be expressed as:
{
\begin{equation}
    \dot{\mathcal{W}} = -T(\mathcal{W} A + A' \mathcal{W} + \mathcal{W} Q \mathcal{W}-\eta_i G).
\end{equation}
Therefore, if we define the Lyapunov Function $V = \Sigma_H'\mathcal{W}\Sigma_H$, we have that:
\begin{equation}
\begin{aligned}
    \frac{d}{dq}\left(\Sigma_H'\mathcal{W}\Sigma_H\right)&=\Sigma_H'\left(T\mathcal{W}A+TA'\mathcal{W}+T\eta_i G+\dot{\mathcal{W}}\right)\Sigma_H\\
    &=-T\Sigma_H'\mathcal{W}Q\mathcal{W}\Sigma_H.
\end{aligned}
\end{equation}
By integrating the previous relation, we have
\begin{multline}
    \Sigma_H'(1)\mathcal{W}(1)\Sigma_H(1)-\Sigma_H(0)\mathcal{W}(0)\Sigma_H(0)=\\-T\int_0^1 \Phi(q,0)'\mathcal{W}Q\mathcal{W}\Phi(q,0) \dq,
\end{multline}
where $\Phi(q_1,q_2)$ is the transition matrix of the system \eqref{eq:homogeneous_derivative_lyapunov_eq} betwen times $q_1$ and $q_2$. Moreover, since $\bar{\Omega}_i(q)$ is periodic with period one and $\Sigma_H(0)=I$, we have that
\begin{multline}
    \label{eq:sigma_1_minus_sigma_0}
    \Sigma_H'(1)\mathcal{W}(0)\Sigma_H(1)-\mathcal{W}(0)\\=-T\int_0^1 \Phi(q,0)'\mathcal{W}Q\mathcal{W}\Phi(q,0) \dq.
\end{multline}
Note that $\mathcal{W}\Phi(q,0)$ is full rank on a nontrivial set, since $\mathcal{W}$ is positive definite and $\Phi(q,0)$ is full rank for at least a non-degenerate interval due to Assumption 1 and the fact that target $i$ is observed at least once in an period. This, along with the fact that $Q$ is positive definite, implies that the integral in \eqref{eq:sigma_1_minus_sigma_0} will be a positive definite matrix. Therefore,
\begin{equation}
    \Sigma_H'(1)\mathcal{W}(0)\Sigma_H(1)-\mathcal{W}(0) \prec 0.
\end{equation}
{Consequently, one can see that
\begin{equation}
\label{eq:induced_norm_sigma}
\frac{(\Sigma_H(1)x)'\mathcal{W}(0)(\Sigma_H(1)x)}{x'\mathcal{W}(0)x}<1,
\end{equation}
for every nonzero $x$. Since $\mathcal{W}(0)$ is positive definite, \eqref{eq:induced_norm_sigma} shows that the norm of the matrix $\Sigma_H(1)$ induced by $\mathcal{W}(0)$ (i.e., $\norm{\Sigma_H(1)}_\mathcal{W}(0)$) is less than 1, therefore its spectral radius is smaller than 1. 
This implies that the absolute value of all the eigenvalues of $\Sigma_H(1)$ are smaller than 1. Hence, $\Sigma_H(1)\otimes \Sigma_H(1)$ is stable, and $\Lambda=\tilde{\Lambda}$.} Moreover, \eqref{eq:discrete_lyapunov_equation} has one solution given by
\begin{equation}
    \label{eq:explicit_solution_lyapunov_equation}
    \Lambda = \sum_{j=1}^\infty\left(\Sigma_H(1)\right)^j\Sigma_{ZI}(1)\left(\Sigma_H(1)'\right)^j.
\end{equation}
We point out that the sum in \eqref{eq:explicit_solution_lyapunov_equation} converges, since the absolute value of the eigenvalues of $\Sigma_H(1)$ are all lower than 1. }

{Now, note that \eqref{eq:derivative_omega_i} is a first order linear matrix differential equation and its general solution is given by
\begin{equation}
    \label{eq:generic_solution_matrix_ode}
    \Sigma(q) = \Sigma_H'(q)\Sigma(0)\Sigma_H(q) + \Sigma_{ZI}(q).
\end{equation}
Since there is a unique solution to \eqref{eq:generic_solution_matrix_ode}, and when ${\partial {\bar\Omega}_i(q)}/{\partial \theta_d}$ exists it must satisfy \eqref{eq:generic_solution_matrix_ode}, we know that  $\Sigma(q)={\partial {\bar\Omega}_i(q)}/{\partial \theta_d}$.}
\end{proof}

{The usefulness of Prop. \ref{prop:solution_lyapunov_eq} for persistent monitoring applications is contingent on the existence of the derivatives $\Sigma(q)={\partial {\bar\Omega}_i(q)}/{\partial \theta_d}$. In Appendix \ref{ap:existence_derivatives} we discuss the existence of these derivatives and show that they indeed exist in most practical situations.}

Also, note that the Lyapunov equation in \eqref{eq:homogeneous_derivative_lyapunov_eq} can be efficiently solved for low-dimensional systems using the algorithm proposed in \cite{barraud1977numerical} and implemented in the MATLAB function $dlyap$. We also highlight that, in order to compute the gradient, the partial derivatives of the steady state covariance matrices must be computed using the procedure in Prop. \ref{prop:solution_lyapunov_eq}. Then, these partial derivatives are used along with \eqref{eq:derivative_steady_state_cost} to compute the partial derivatives of the cost, which compose the gradient $\nabla J$. Algorithm \ref{alg:infinite_horizon_optimization} summarizes the procedure to compute the steady state gradients.

\begin{algorithm}
\caption{Steady State Gradient Computation}
\label{alg:infinite_horizon_optimization}
\begin{algorithmic}[1]
\Procedure {ComputeSteadyStateGradient}{}
\State{\bf Input}: $\Theta$
\State Compute $s_1(q),...,s_N(q)$ from the parameterization
\For{$i$ ranging from $1$ to $M$}
\State Compute the steady state covariance $\bar{\Omega}_i(q)$
\EndFor
\State{Compute $\frac{\partial}{\partial \theta}\int_0^{t_f}\sum_{j=1}^N{u}_j'(\zeta){u}_j(\zeta)d\zeta$ according to the parameterization}
\State{Compute $\frac{\partial s_j(t)}{\partial \theta}$ and $\frac{\partial T}{\partial \theta}$ according to the parameterization}
\For{every $\theta$ in $\Theta$}
\For {$i$ ranging from $1$ to $M$}
\State Compute $\frac{\partial \Omega_i(q)}{\partial \theta}$ as indicated in Prop. \ref{prop:solution_lyapunov_eq}.
\EndFor
\State Compute $\frac{\partial J}{\partial \theta}$ using \eqref{eq:derivative_steady_state_cost}
\EndFor
\State {\bf Output:} $\nabla J$
\EndProcedure
\end{algorithmic}
\end{algorithm}

In order to locally optimize the trajectories, the gradient computation needs to be used along with some gradient descent scheme. We describe the optimization procedure we used in Alg. \ref{alg:gradient_descent}, 
where $\kappa_l$ is a scalar positive gain, and the $proj$ operator projects the parameters into the set of feasible parameters ($u_j(t)\in\mathcal{U}$). As a side note, this projection might be difficult to compute in general and, therefore when choosing a parameterization it is important to make sure that there are efficient ways to compute this projection numerically.

 \begin{algorithm}
 \caption{Gradient Descent}
 \label{alg:gradient_descent}
 \begin{algorithmic}[1]
 \Procedure{Gradient Descent}{}
\State {\bf Input:} $\Theta^0$, 
\State $||\nabla J|| \leftarrow \infty$
\State $l\leftarrow 0$
\While {$||\nabla J||>\epsilon$}
\State $\nabla J\leftarrow$ComputeGradient($\Theta^l$)
\State ${\Theta}^{l+1}\leftarrow \text{proj}(\Theta^{l}-\kappa_l\nabla J)$
\
\State $l\leftarrow l+1$

\EndWhile
\State {\bf Output:} $\underline{\Theta}^l$
\EndProcedure
\end{algorithmic}
 \end{algorithm}


\section{Parameterization of an Optimal Trajectory in 1-D with speed bounds}
\label{sec:parameterizations}
When the agents and targets are constrained to a line, a particularly interesting case is the one where the absolute value of controls is bounded ($\mathcal{U}=\{u\in \mathbb{R}\ |\ |u|<u_{\max}\}$) and there is no penalty for control effort in the optimization cost $J$ (i.e. $\beta = 0$). In this case we can represent optimized controls using a simple parameterization that could even lead to global optimality. It is worth noticing that in many real-world applications of persistent monitoring agents are constrained to (possibly multiple) uni-dimensional mobility paths, such as powerline inspection agents, cars on streets, and autonomous vehicles in rivers.

Assuming proper rescaling, we can consider $-1\leq u_j \leq 1$, i.e., $\mathcal{U}=[-1,1]$. In the remainder of this section, we derive properties of the optimal control, establish a parameterization that is able to represent an optimal control, and then compute the gradients necessary in order to optimize the trajectories.

\subsection{Properties of an Optimal Control}
In order to derive the properties of an optimal control,
we first introduce the following lemma.  The intuition behind it is that if a target is observed for a longer time (or with better quality), its uncertainty will be lower. {We note that, although this lemma is introduced in this Section, it is not restricted to the 1D setting with bounded input.}

\begin{lemma}
\label{prop:monotonicity_ricatti_eq}
 Given $\Omega_1(t)$ and $\Omega_2(t)$, two bounded covariance matrices under the dynamics in \eqref{eq:dynamics_omega} with $A=A_1=A_2$, $G=G_1=G_2$, $Q=Q_1=Q_2$, then if $\Omega_1(0)-\Omega_2(0)$ is negative semi-definite and $\eta_1(t) \geq \eta_2(t)\ \forall t$, then $\Omega_1(t)-\Omega_2(t)$ is a negative semi definite matrix for all $t\geq 0$.
\end{lemma}
\begin{proof}
    Define $\beta=\Omega_1(t)-\Omega_2(t)$. The dynamics of $\beta$ is described by the following equation.
    \begin{multline}
        \label{eq:dynamics_xi_initial}
        \dot{\beta}(t) = A\beta(t)+\beta A'-\eta_1(t)\Omega_1(t)G\Omega_1(t)\\+\eta_2(t)\Omega_2(t)G\Omega_2(t).
    \end{multline}
    Adding and subtracting the terms $\eta_1(t)\Omega_2(t)G\Omega_2(t)$ and $\eta_1(t)\Omega_1(t)G\Omega_2(t)$ to the equation, we can rewrite \eqref{eq:dynamics_xi_initial} as:
    \begin{multline}
        \label{eq:dynamics_xi}
        \dot{\beta}(t) = A\beta(t)+\beta A'-\eta_1(t)\left[\Omega_1(t)G\beta(t)+\beta(t)G\Omega_2(t)\right]\\+\left[\eta_2(t)-\eta_1(t)\right]\Omega_2(t)G\Omega_2(t).
    \end{multline}
    From Thm.~1.e in \cite{kriegl2011denjoy}, since $\beta(t)$ is a $C^1$ matrix, its eigenvalues 
can be $C^1$ time parameterized. Let $\mu_n$ denote the $n^{th}$ eigenvalue of $\beta(t)$ and $x_n(t)$ the corresponding unit norm eigenvector. 
Then, from Thm.~5 in \cite{lancaster1964eigenvalues} we have that
    \begin{equation*}
        \dot{\mu}_n=  x_n'\dot{\beta} x_n.
    \end{equation*}
Also, notice that by using \eqref{eq:dynamics_xi} and the fact that $\lambda_{\min}\left(\frac{D+D'}{2}\right)\leq \frac{x'Dx}{\norm{x}}\leq \lambda_{\max}\left(\frac{D+D'}{2}\right)=\norm{D}$, for any square matrix D,
\begin{align*}       
    \dot{\mu}_n &\leq \norm{A}\mu_n-\eta_1\beta\mu_n+\left[\eta_2-\eta_1\right]x_n'\Omega_2G\Omega_2x_n \\
    &\leq \norm{A}\mu_n-\eta_1\beta\mu_n,
\end{align*}
where $\beta=\lambda_{min}\left((\Omega_1+\Omega_2)G+G(\Omega_1+\Omega_2)\right)$.
Using Gronwall's inequality \cite{gronwall1919note} and the fact that the solution of a first order linear homogeneous ODE does not change sign, we conclude that $\mu_n(t)\leq0, \forall\ t\in [0,T]$ and, therefore, $\beta(t)$ is negative semidefinite.
\end{proof}
In Lemma \ref{prop:monotonicity_ricatti_eq}, $\Omega_1$ and $\Omega_2$ can also be understood as covariance matrices for the same target but under different agent trajectories.

Before proceeding to the proposition about an optimal control structure, a few definitions are necessary. We define an {\it isolated target} $i$ as a target such that 
\begin{align*}
	\min\limits_{k\neq i}|x_i-x_k|>2r_{\max}, \quad r_{\max}=\max_{i,j}\{r_{i,j}\}.
\end{align*} 
Therefore, an isolated target is a target for which an agent cannot see another target when visiting it. Referring to the regions in space where an agent can sense a target as ``visible areas", the minimum distance between visible areas $d_{\min}$ is defined as:
\begin{align*}
	d_{\min} = \min_{i,k}|x_i-x_k|-2r_{\max}>0,
\end{align*}
and the finite time cost is defined as
\begin{equation}
    \label{eq:finite_time_opt_objective}
    J(u_1,...,u_N,t)=\frac{1}{t}\int_0^{t}\left(\sum_{i=1}^M\text{tr}\left(\Omega_i(\beta)\right)\right) d\beta.
\end{equation}

We can then claim the following proposition.

\begin{proposition}
    \label{prop:optimal_control}
    In an environment where all the targets are isolated, given any policy $u_j(\beta)$, $j=1,...,N$, then there is a policy $\tilde{u}_j(\beta)$ with $\tilde{u}_j(\beta)\in\{-1,0,1\}$ $\forall \beta\in[0,t]$ and with the number of control switches for each agent (i.e. discontinuities in $\tilde{u}_j(\beta)$) upper bounded by $2\frac{t}{d_{\min}}+4$ such that 
    $J(u_1,...,u_N,t)\geq J(\tilde{u}_1,...,\tilde{u}_N,t)$.
\end{proposition}
\begin{proof}
    We prove this result by construction: given a policy $u_j(t')$ with $\eta_i(t')$ associated to it (as defined by \eqref{eq:definition_eta}), we will construct an alternative policy $\tilde{u}_j(t')$ associated with $\tilde{\eta}_i(t')$ such that $\tilde{\eta}_i(t') \geq \eta_i(t')\ \forall t'\in[0,t]$ and $i=1,...,M$, and then use Prop. \ref{prop:monotonicity_ricatti_eq}, along with the definition of the cost \eqref{eq:finite_time_opt_objective}, to show that the alternative policy has lower or equal cost than the original one.
    
    Initially, we focus on the policy $u_j(t')$.  We say that an agent $j$ ``visits" a target $i$ if at some time $t'$, $|s_j(t')-x_i(t')|<r_j$. For every agent in the policy $u_j(t')$, there is an ordered collection of targets it visits in $[0,t]$. Therefore, there must exist a set of indices of all the targets visited by agent $j$: $\{y^j_0,...,y^j_{K_j}\}\in \{1,...,M\}$, such that $y_p^j\neq y_{p-1}^j$ and agent $j$ visited no other target in the time between visiting targets $y_p^j$ and $y_{p-1}^j$. This is the sequence of all the targets that agent $j$ visited over $[0,t]$, not considering consecutive visits to the same target. In other words, the same target can be present more than once in the sequence $\{y^j_0,...,y^j_{K_j}\}$ but, if that is the case, it will not be in consecutive positions.
    
    For each of these visits, we can define the initial visiting time $t_p^j$ for $p=1,...,K_j$ as
    \begin{align*}
        t_p^j = \inf \{t' |t'>t_{p-1}^j \text{ and agent $j$ visits target $y_p^j$ at time $t'$} \}, 
    \end{align*}
    and $t_0^j=0$ and $t_{K_j+1}^j=t$. Also note that while $t' \in [t_{p-1}^j,t_p^j)$, agent $j$ only influences the value of $\eta_i(j)$ of the target it is currently visiting.
    We propose the following alternative policy, where $\tilde{u}_j(t')$ for $t'\in[t_{p-1}^j,t_{p}^j)$ is such that:
    \begin{equation*}
        \label{eq:alternative_u}
        \tilde{u}_j(t')=\begin{cases}
            \frac{s_j(t_p^j)-s_j(t')}{|s_j(t_p^j)-s_j(t')|},&\text{ if }\frac{|s_j(t_p^j)-s_j(t')|}{t_p^j-t'} \leq 1,\\
            \frac{x_{y_p^j}-s_j(t')}{|x_{y_p^j}-s_j(t')|},&\text{ if }\frac{|s_j(t_p^j)-s_j(t')|}{t_p^j-t'}>1 \text{ and }\\&\qquad \qquad \ \ \  \ s_j(t') \neq x_{y_p^j}.\\
            0,&\text{ otherwise.}\ \\ 
        \end{cases}
    \end{equation*}
    Notice that this construction provides a feasible trajectory, since the original trajectory is assumed feasible. Also, in the alternative policy $\tilde{u}_j(t')\in\{-1,0,1\}$ $\forall t'\in[0,t]$, since the speed is either zero or a scalar divided by its absolute value.
    
    The intuition behind the proposed alternative policy is that at the beginning of each visit, the agent moves with maximum speed towards the target $y_p^j$ and if it reaches the target, it dwells on top of it. However, it must move in a way such that it begins the next visit at the same time as in the original policy, i.e., the positions of agent $j$ associated to the alternative policy $\tilde{s}_j(t')$ is such that $\tilde{s}_j(t_p^j)={s}_j(t_p^j)$.
    
     Also, for time $t'\in [t_{p}^j,t_{p+1}^j]$ both the original and the alternative policies only influence the value of $\eta_i$ for $i=y_p^j$, since in the alternative policy the agent is closer (or at least as close) to the currently visited target. Thus, from \eqref{eq:definition_eta} we have that
    \begin{equation*}
        \tilde{\eta}_i(t')\geq \eta_i(t'),\ \forall t'\in[0,t],\ i\in\{1,...,M\}.
    \end{equation*}
    Therefore, using Lemma \ref{prop:monotonicity_ricatti_eq} and the cost definition \eqref{eq:finite_time_opt_objective}, we get that
     \begin{multline*}
    J(\tilde{u}_1,...,\tilde{u}_N,t)-J({u}_1,...,{u}_N,t)=\\
    \frac{1}{t}\int_0^t\sum_{i=1}^M\tr{\left(\tilde{\Omega}_i(t')-{\Omega}_i(t')\right)}\leq0.
     \end{multline*}
     which shows that the alternative policy has a lower or equal cost compared to the original one. Note that, due to velocity constraints, in both the original and the alternative policy there is a maximum of $\frac{t}{d_{\min}}+1$ visits to targets per agent. Moreover, in the alternative policy, an agent has at most 2 velocity switches at each target visit. Therefore, at most $2\frac{t}{d_{\min}}+4$ velocity switches can happen due to target visits, plus one switch to match the initial position of the original policy and another to match the terminal position of the original policy. 
     This implies that the maximum number of velocity switches in the alternative policy is $2\frac{t}{d_{\min}}+4$. 
\end{proof}
One way to interpret this proposition is that if one looks ahead at the next $T$ units of time (where $T$ is the period of a periodic solution or the prediction horizon, in the transient case), any trajectory can be improved (or at least, maintain same cost) by adequately selecting its controls $u_j(t)$ in the set $\{-1,0,1\}$. Also, notice that even though we were not able so far to prove that the same result holds when the targets are not necessarily isolated, the same structure can still be used but without the guarantee of optimality.

\subsection{Parameterization of an Optimal 1D Trajectory}
The result in Prop. \ref{prop:optimal_control} implies that when the targets are isolated, {there is no loss of performance if we restrict ourselves to controls of the form $u_j(t)\in\{-1,0,1\}\ \forall t>0$, with a bounded number of control switches}. This property allows the optimal trajectory to be described by a finite set of parameters, similar to optimal control results in previous work by the authors \cite{cassandras2013optimal,pinto2019monitoring}. Here, in particular, we are looking into periodic trajectories and, hence, this property implies that the movement in each period of agent $j$ consists of a sequence of dwelling at the same position for some duration of time followed by moving at maximum speed to another location. Therefore, one period of the trajectory of an agent $j$ can fully be described by the following set of parameters:
\begin{enumerate}
\item $T$, the period of the trajectory. 
\item $s_{j}(0)$, the initial position.
\item $\omega_{j,p}$, $p=1,...,P_j$, the normalized dwelling times for agent $j$, i.e., the agent dwells for $\omega_{j,p}T$ units of time before it moves with maximum speed for the $p$-th time in the cycle.
\item $\tau_{j,p}$, $p=1,...,P_j$, the normalized movement times for agent $j$, i.e., the agent $j$ moves for $\tau_{j,p}T$ units of time to the right (if $p$ is odd) or to the left (if $p$ is even) after dwelling for $\omega_{j,p}T$ units of time in the same position.
\end{enumerate}

To enforce consistency of the trajectory, we add the following constraints:
\begin{equation}
\label{eq:natural_parameters_constraints}
\begin{gathered}
    \tau_{j,m}\geq 0,\ \omega_{j,m}\geq 0,\ T\geq 0.
\end{gathered}    
\end{equation}
Notice that this description does not exclude transitions of $u_j$ of the kind $\pm1\rightarrow \mp 1$ and $\pm1\rightarrow 0 \rightarrow \pm1$, since it allows $\omega_{j,m}=0$ and $\tau_{j,m}=0$. In addition to the constraints in \eqref{eq:natural_parameters_constraints}, in order to ensure periodicity, we need to make sure that the sum of the movement times and dwelling times does not exceed one period and that the total time spent moving to the left is equal to the total time spent moving to the right over one period (i.e. the agent returns to its initial position at the end of the period). Therefore, we have the additional constraints:
\begin{equation}
\label{eq:periodic_parameters_constraints}
\begin{gathered}
     \sum_{m=1}^{P_j}(\tau_{j,m}+\omega_{j,m})\leq 1,\ \  \sum_{m=1}^{P_j}(-1)^m \tau_{j,m}=0.
\end{gathered}    
\end{equation}

This parameterization defines a hybrid system in which the dynamics of the agents remain unchanged between events and abruptly switch when an event occurs. Events are given by a change in control value at completion of movement and dwell times. Note that these may occur simultaneously, for instance, if the dwell time is zero (representing a switch of control from $\pm 1$ to $\mp 1$).
This parameterization also applies to the aperiodic transient case, with minor modifications to the constraints imposed to the parameters. Although we do not explore all the details for the sake of readability, we refer the interested reader to \cite{pinto2019monitoring}. 


\subsection{Position Gradients}
Given this parameterization, we use the procedure given in Sec. \ref{sec:transient_opt} to optimize the cost. However, one item missing in Sec. \ref{sec:transient_opt} for computing the gradient of the covariance matrix was the gradient of the agent position with respect to the parameters defining the trajectory.

The movement and dwelling time parameterization defines, along with the uncertainty metric, a hybrid system. For such systems, Infinitesimal Perturbation Analysis (IPA) can be used to compute an event-driven online estimate of the stochastic gradient of the system. An important feature of IPA is that the unbiased gradient estimate can be computed online using only the data observed along the trajectory. Even though we do not discuss in this paper the details of the IPA interpretation of the equations in this subsection, we refer the reader to \cite{cassandras2010perturbation,cassandras2013optimal} for more information about IPA. 

One can see that the position of agent $j$ at normalized time $q$, after the $k$-th event and before the $k+1$-th is
\begin{equation}
s_j(q) -s_j(0)=
\begin{cases}
    T\biggl((-1)^{k/2+1}\biggl(q-\sum_{p=1}^{k/2-1}(\tau_{j,p}+\omega_{j,p})\\ {\ \ }+\omega_{j,\frac{k}{2}}\biggr)+\sum_{p=1}^{k/2}(-1)^{p+1}\tau_p\ \biggl),\ k\ \text{even},\\
    T\sum_{p=1}^{\frac{k-1}{2}}(-1)^{p+1}\tau_{j,p},\ k\ \text{odd}.
\end{cases}
\end{equation}

Therefore, we can compute the following gradients,
\begin{subequations}
\begin{align}
    \frac{\partial s_j}{\partial \tau_{j,m}} &= 
 \begin{cases}
    \left((-1)^{\frac{k}{2}}+(-1)^{m+1}\right)T,&\ m<\frac{k}{2},\ k\ \text{even},\\
    (-1)^{m+1}T,\ m\leq\frac{k-1}{2},&\ k\ \text{odd},\\
    0,&\ \text{otherwise},
 \end{cases}\\
    \frac{\partial s_j}{\partial \omega_{j,m}} &=
    \begin{cases}
        T,\ m\leq\frac{k}{2},&\ k\ \text{even},\\
        0&,\ \text{otherwise},
    \end{cases}\\
    \frac{\partial s_j(q)}{\partial T} &= \frac{s_j(q)-s_j(0)}{T},\\
    \frac{\partial s_j}{\partial s_j(0)}&=1.
\end{align}
\end{subequations}
\subsection{Initial Trajectory for the Optimization}
\label{subsec:initialization_1d}
While we use a gradient descent approach in Alg. \ref{alg:gradient_descent} to locally minimize the cost function, it is necessary to find an initial parameter configuration. Therefore, we propose a method to efficiently compute a starting point for the optimization.

Proposition \ref{prop:unique_attractive_sol_riccati_eq} states that if every target is visited at least once in a periodic trajectory, then the steady-state covariance matrix exists. However, if in a periodic trajectory one of the targets is never visited and its internal state dynamics is unstable, then the estimation error will grow without bound as time goes to infinity. 
Also, when a target is not visited in the initial trajectory, the gradient descent optimization may converge to undesired solutions, a problem known as the ``lack of event excitation" and discussed in depth in \cite{Khazaeni:2016jv}. Therefore, this kind of initial trajectory will not be considered in this work.

In this section, we discuss a method for finding these initial trajectories that will always lead to a feasible initial configuration. Note that due to the local nature of our optimization procedure, different initial conditions can lead to different local optima. We, therefore, leverage intuition about the problem to provide reasonable initial solutions with the hope that they will converge to good local optima.

The idea of finding a schedule where all the targets are visited fits naturally into a graph search paradigm, where the targets are modelled as nodes and the edge weights between nodes are the distances between the targets. The problem of finding a feasible schedule can be translated to one of finding $N$ sequences (that represent the schedule of each agent) of nodes where each target belongs to at least one of these sequences. One can add to that a cost function that guides the way in which these sequences are created. A goal that intuitively will lead to reasonable initial solutions is to minimize the distance of the agent that has the longest travel path. This is the well known MTSP (see \cite{BEKTAS2006209} for a good overview of this problem and approaches to solve it). It is worth mentioning that the MTSP is NP-hard, and, therefore, intractable. However, meta-heuristic approaches can provide feasible, though not necessarily optimal, solutions. In this work, we use the genetic algorithm described in \cite{TANG2000267} to find heuristic solutions. This approach is interesting because it finds a feasible solution in the first iteration and refines it as the number of iterations increases. Therefore, one can decide how much computation time to spend, leveraging the tradeoff between optimality and computation effort spent in generating this initial trajectory.

The MTSP problem finds a minimal length cycle and therefore can be immediately converted to parameters that represent one period of the steady state solution. We choose the dwelling times to be initially zero. 

\subsection{1D Simulation Results}
In the simulations, we have chosen to highlight interesting aspects of the solution, rather than simply give an example of the techniques discussed in this paper. We have analyzed a steady state problem with 2 agents and 5 targets. We used the following matrices in the state evolution model
\begin{align*}
    A_i = \begin{bmatrix} -1 & -0.1 \\ -0.1 & 0.01 \end{bmatrix}, \quad Q_i = \textrm{diag}(1,1),
\end{align*}
and the following parameters for the observation model
\begin{align*}
    H_i = \textrm{diag}(1,1), \quad R_i = \textrm{diag}(1,1),\quad r_j=0.9.
\end{align*}
Instead of using the initialization method proposed in Subsec. \ref{subsec:initialization_1d}, we used the following set of parameters:
\begin{align*}
    s_1^0(0)=2.7,\quad s_2(0)=6.8, \quad T^0=6, 
\end{align*}
\begin{align*}
    \tau_1^0=\tau_2^0=0.1[1,0.1,1,1,0.1,1,0.1,1,1,0.1,1],
\end{align*}
\begin{align*}
    \omega_1^0=\omega_2^0=0.0125[1,1,1,1,1,1,1,1,1,1,1].
\end{align*}

The goal of using these initialization parameters was to have both agents share one target in the first iteration of the optimization process and then explore whether or not they would remain sharing the target after the local optimization procedure. The gradient descent step size was set to be constant, $\kappa_0=\kappa_l=0.02$. 

\begin{figure*}[htp!]
    \centering
    \begin{subfigure}[t]{0.32\textwidth}
        \centering\includegraphics[width=\textwidth]{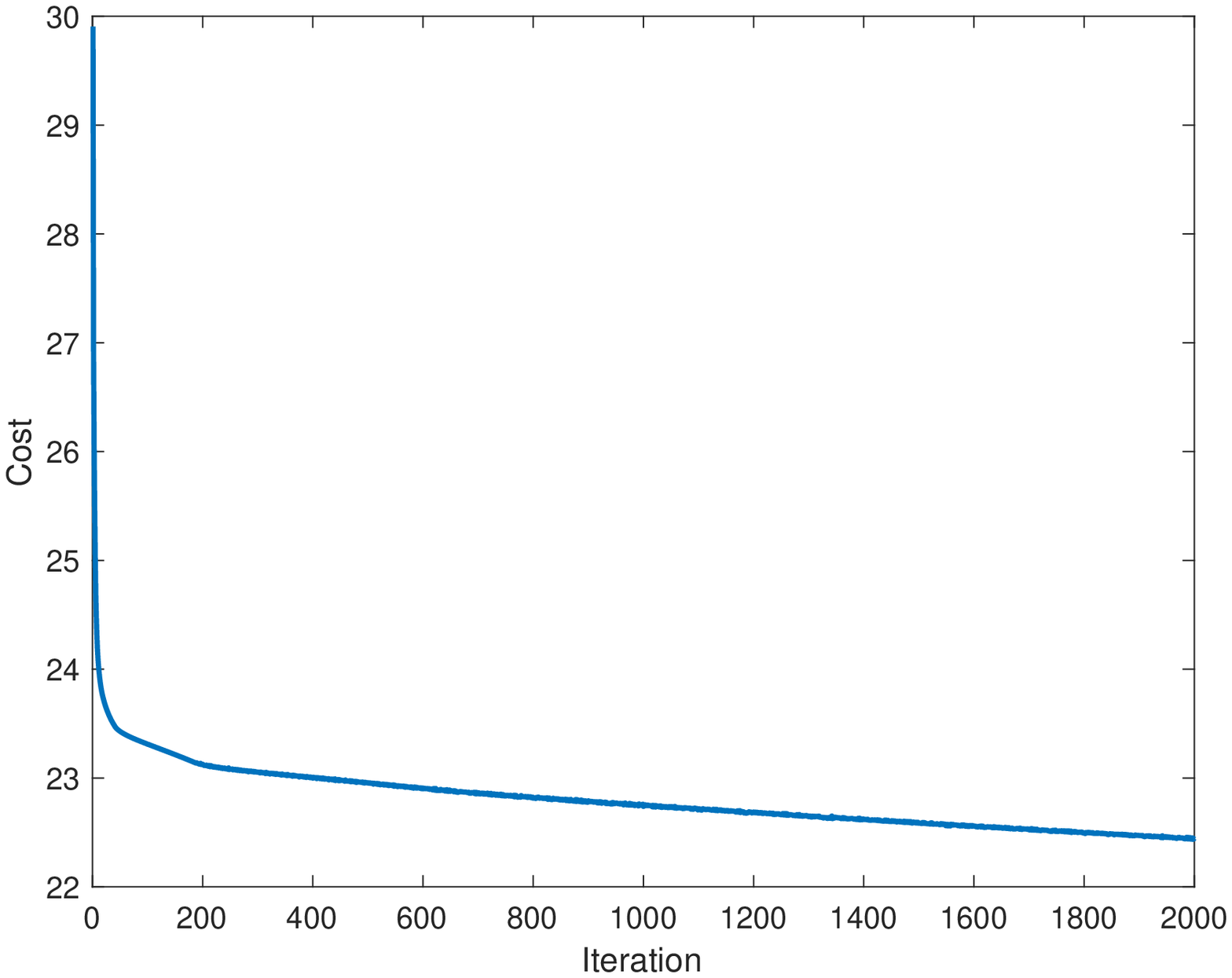}
        \caption{Cost vs. iteration number}
        \label{fig:cost_5_target}
    \end{subfigure}
    \begin{subfigure}[t]{0.32\textwidth}
        \centering\includegraphics[width=\textwidth]{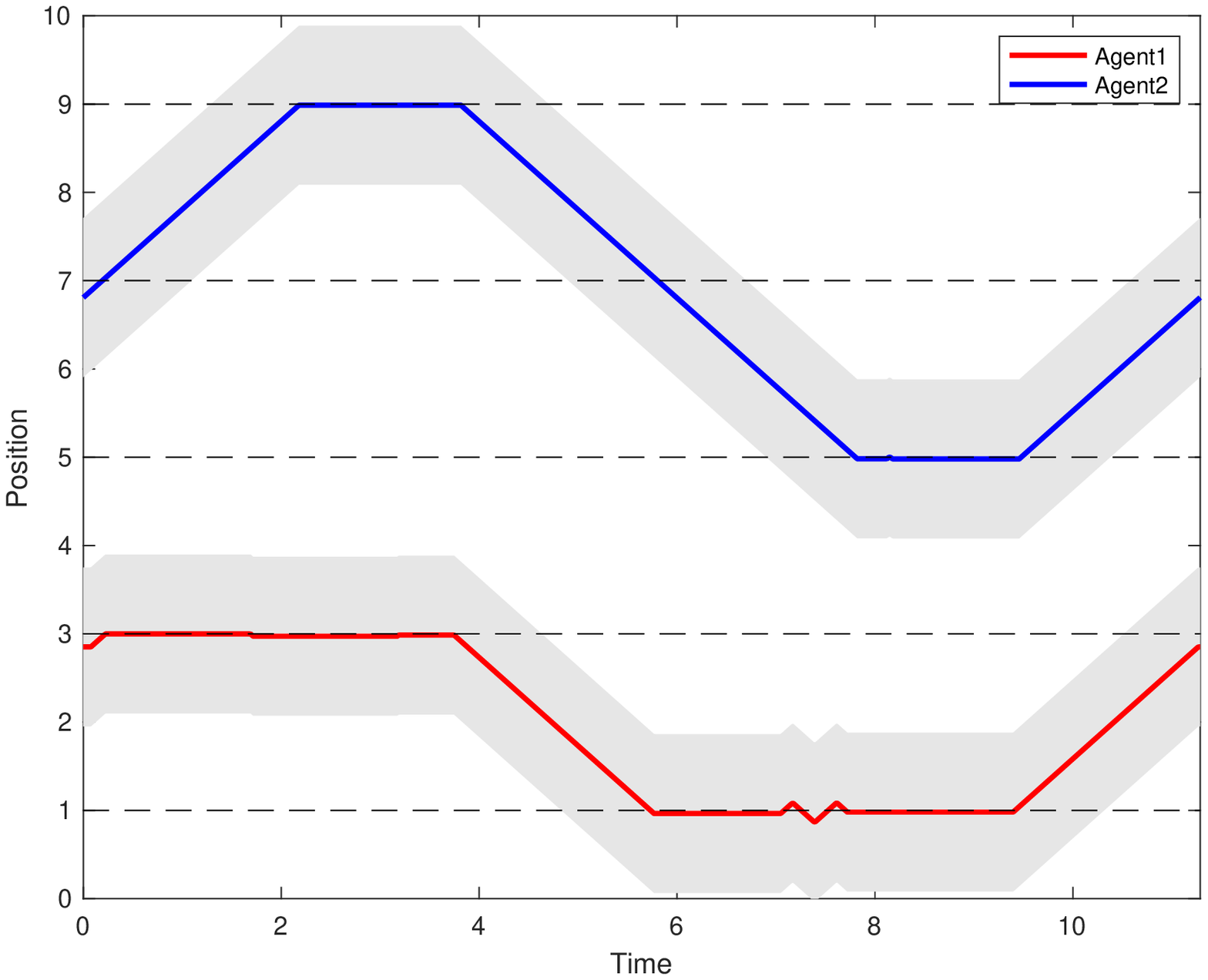}
        \caption{Agent trajectories at final iteration}
        \label{fig:position_5_target}
    \end{subfigure}
    \begin{subfigure}[t]{0.32\textwidth}
        \centering\includegraphics[width=\textwidth]{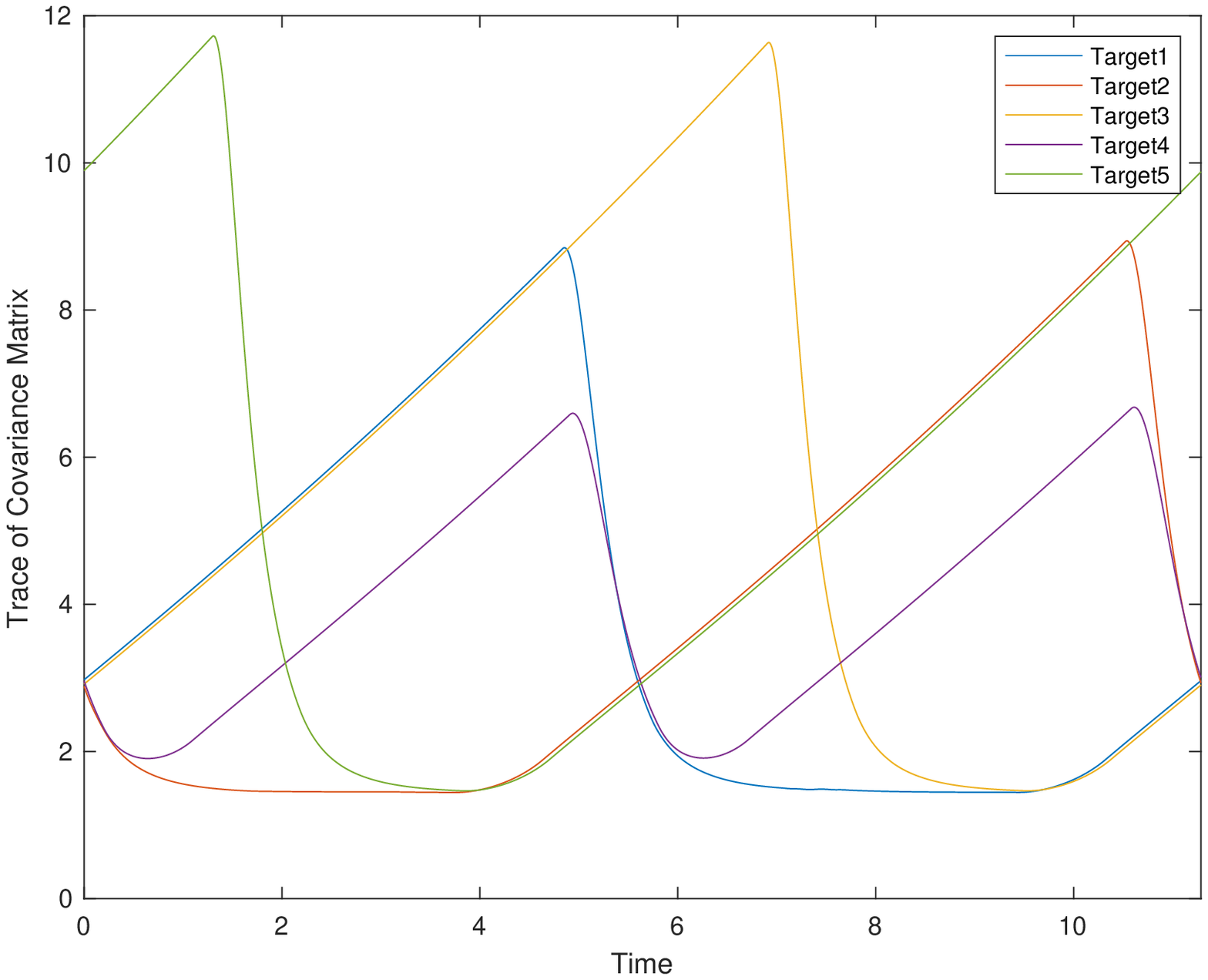}
        \caption{Trace of the covariance for each target}
        \label{fig:covariance_5_target}
    \end{subfigure}
    \caption{Results of a simulation with two agents and five targets. (a) Evolution of the overall cost as a function of iteration number on the gradient descent. (b) Trajectories of the agents at the final iteration. The dashed lines indicate the positions of the targets and the grey shaded area the visibility region of the agent. (c) Evolution of the trace of the estimation covariance matrices of the five targets. }
    \label{fig:results_5_target}
\end{figure*}

Figure \ref{fig:results_5_target} shows the results of the optimization in this scenario. Notice that even though both agents and all the targets have the same dynamic models, the solution at the last iteration of the optimization was such that one of the agents visits three of the targets and the other two of them. {One interesting aspect of the trajectories of the targets in Fig. \ref{fig:position_5_target} is that in the period between times 6 and 8 agent 1 makes a movement with small amplitude around target 1. The effects of this oscillatory movement are hard to notice in the trace of the covariance of target 1 in Fig. \ref{fig:covariance_5_target}, which implies that the difference in performance is negligible. Therefore, even though it is intuitively clear that staying still rather than moving with this oscillatory behavior will lead to a lower cost solution, the difference in terms of cost is minor. Also, notice that the solution has not yet fully converged, as can be seen in Fig. \ref{fig:cost_5_target} and further iterations would remove this small oscillatory behavior.} 

Finally, we point out that while the maximum number of switches in a direction allowed to each agent was set to 11, the final solution appears to have fewer because some of the movement and dwelling times in the final solution are zero.

\section{Fourier Curves for Multi-dimensional Persistent Monitoring with Unbounded Speed}
\label{sec:fourier_curves}
For the 1D case we derived a parameterization with a finite number of parameters of the optimal solution. Unfortunately, the same result does not extend to multi-dimensional persistent monitoring problems. Therefore, instead of looking for an exact representation of the optimal trajectory, we focus on a family of parameterized curves that can approximate very general curves. We pick as an illustration the case where speed is not bounded, in part because the projection operation in line 7 of Alg. \ref{alg:gradient_descent} becomes trivial.
Note that whenever the constant that weights the control effort penalization is not zero, i.e. $\beta\neq0$ as defined in \eqref{eq:cost_function_T}, the fact that the control effort is considered in the total cost 
will not allow the control to be unbounded. An appropriate choice of $\beta$ can provide adequate speed bounds for any given dynamics of the system. As a side note, we highlight that bounded speeds can also be handled in this framework, however the projection operator in the gradient descent optimization becomes more complex. 

Since periodicity is an  essential feature of the steady-state analysis discussed in this work, a natural choice is to use a truncated Fourier series to represent the movement of the agents in each of the coordinates $e_p$, $p=1,...,P$, i.e.
\begin{equation}
    \label{eq:def_fourier_curves}
    {s}_j^{e_p}(q)=s_{j,0}^{e_p}+\sum_{k=1}^K a_{j,k}^{e_p} \sin(2\pi f_k q)+b_{j,k}^{e_p}(\cos(2\pi f_kq)-1),
\end{equation}
where $f_k$ are integer frequencies and, therefore, $s_j^{e_p}(q)$ is periodic with period 1.
The set of parameters that fully characterize all the agents trajectories is $\Theta = \{\{a_{j,k}^{e_p}\},\{b_{j,k}^{e_p}\},\{s_{j,0}^{e_p}\},T\}$, $j=1,..,N$, $p=1,...,P$, $k=1,...,K$. As in the 1D case, in order to compute the derivative of the covariance matrix, we need to give a procedure to compute $\frac{\partial s_k}{\partial \theta}$. For any parameter $\theta\in\Theta$,

\begin{subequations}
\label{eq:trajectory_derivatives_fourier_curves}
\begin{align}
    \frac{\partial s^{e_p}_j}{\partial a_{m,k}^{e_r}} &= \begin{cases}
    \sin(2\pi f_k q),&\ \text{if }j=m\text{ and }p=r,\\
    0,&\ \text{otherwise},
    \end{cases}\\
    \frac{\partial s^{e_p}_j}{\partial b_{m,k}^{e_r}} &= \begin{cases}
    \cos(2\pi f_k q)-1,&\ \text{if }j=m\text{ and }p=r,\\
    0,&\ \text{otherwise},
    \end{cases}\\
    \frac{\partial s^{e_p}_j}{\partial s_{m,0}^{e_r}} &= \begin{cases}
    1,&\ \text{if }j=m\text{ and }p=r,\\
    0,&\ \text{otherwise},
    \end{cases}\\
    \frac{\partial s^{e_p}_j}{\partial T} &= 0.
\end{align}
\end{subequations}

The derivatives in \eqref{eq:trajectory_derivatives_fourier_curves} give enough information to compute the partial derivatives of the steady state covariance matrix as indicated in Prop. \ref{prop:solution_lyapunov_eq}.
In order to compute the gradient of the cost function, the following expression can be used: 
\begin{equation}
    \label{eq:complete_derivative_cost}
    \frac{\partial J}{\partial \theta} = \int_0^1 \sum_{i=1}^N\tr\left(\frac{\partial\Omega_i}{\partial \theta}\right)dq+\beta\frac{\partial}{\partial \theta}\sum_{j=1}^N\int_0^1\norm{\frac{ds_j}{dt}}^2dq.
\end{equation}
Note that
\begin{equation}
    \frac{ds_j}{dq} = T\frac{ds_j}{dt}.
\end{equation}
Using \eqref{eq:def_fourier_curves}, we can compute
\begin{multline}
    \sum_{j=1}^N\int_0^1\norm{\frac{ds_j}{dt}}^2dq=\\\sum_{j=1}^N\sum_{p=1}^P\sum_{k=1}^K\frac{(2\pi f_k)^2}{2T^2}\left(\left(a_{j,k}^{e_p}\right)^2+\left(b_{j,k}^{e_p}\right)^2\right),
\end{multline}
and, therefore,
\begin{subequations}
\label{eq:derivative_total_power}
\begin{align}
    \frac{\partial}{\partial a_{j,k}^{e_p}}\sum_{j=1}^N\int_0^1\norm{\frac{ds_j}{dt}}^2dq&=
    \frac{(2\pi f_k)^2}{2T^2}a_{j,k}^{e_p},\\
    \frac{\partial}{\partial b_{j,k}^{e_p}}\sum_{j=1}^N\int_0^1\norm{\frac{ds_j}{dt}}^2dq&=
    \frac{(2\pi f_k)^2}{2T^2}b_{j,k}^{e_p},\\
    \frac{\partial}{\partial s_{j,0}^{e_p}}\sum_{j=1}^N\int_0^1\norm{\frac{ds_j}{dt}}^2dq&=
    0,
\end{align}
\begin{multline}
    \frac{\partial}{\partial T}\sum_{j=1}^N\int_0^1\norm{\frac{ds_j}{dt}}^2dq=\\
    \sum_{j=1}^N\sum_{p=1}^P\sum_{k=1}^K\frac{-(2\pi f_k)^2}{T^3}\left(\left(a_{j,k}^{e_p}\right)^2+\left(b_{j,k}^{e_p}\right)^2\right).
\end{multline}
\end{subequations}

\subsection{Optimization Initialization}
In the multi-dimensinal optimization, we still use the suboptimal solution of the MTSP problem as a starting point. However, unlike the 1-D scenario with the movement and dwelling time parameterization, the heuristic solution of the MTSP problem cannot be directly converted to a Fourier Curve trajectory.  The solution of the MTSP problem gives, for each agent $j$, a cyclic schedule of targets $\mathcal{S}_j=\{y_j^1,...,y_j^{Y_j},y_j^1\}$ and, therefore, it is still necessary to obtain the parameters $\Theta = \{\{a_{j,k}^{e_p}\},\{b_{j,k}^{e_p}\},\{s_{j,0}^{e_p}\},T\}$ from this schedule. We define $d_j^m$ as the cumulative distance that the agent has traveled when it reaches the $m$-th target in the schedule $\mathcal{S}_j$, and $D_j$ as the total distance traveled by an agent in one cycle. We then look for a feasible truncated Fourier series trajectory such that at the normalized time $q=d_j^m/(D_jT)$, the agent is at a distance lower or equal to the sensing radius (multiplied by a factor $1-\delta$, $0<\delta<1$, in order to give some distance margin) from the target. The position of the agent at the beginning of the cycle is set to be the position of the first target in the schedule $\mathcal{S}_j$.
The period $T$ can be set to any positive number. For each of the agents, the following optimization problem gives a set of feasible $\{a_{j,k}^{e_p}\},\{b_{j,k}^{e_p}\}$.

\begin{mini}  
     {a_{j,k}^{e_p},b_{j,k}^{e_p}}{\sum_{p=1}^P\sum_{k=1}^Kf_k|a_{j,k}^{e_p}|+f_k|b_{j,k}^{e_p}|}{}{ \label{eq:optimization_initial_solution}}
  \addConstraint{\norm{s_j\left(\frac{d_j^m}{D_j}\right)-x_{y_j^m}}_2}{\leq(1-\delta)r_{i,j},}{\ m=1,..,Y_j}
\end{mini}

Note that if we substitute the definition \eqref{eq:def_fourier_curves} into the constraint \eqref{eq:optimization_initial_solution}, this optimization can be formulated as a Quadratically Constrained Program, which is a convex optimization problem that can be solved efficiently. From our experience, minimizing a weighted sum of absolute values in the objective function of \eqref{eq:optimization_initial_solution} has led to smooth initial trajectories. However, other optimization objectives could be used.

It is worth observing that for each of the agents, the trajectory generated by the heuristic solution of the MTSP problem consists of segments of straight lines that visit each of the targets in the schedule $\mathcal{S}_j$. Note that this trajectory, as a function of time, composed by sequence of straight lines can be projected in each of the axis $e_p$ and the projection in that axis will still be a sequence of segments of straight lines. Since piecewise linear functions can be represented by Fourier series, there always exist a $K$ large enough such that there is a solution to \eqref{eq:optimization_initial_solution} because for that $K$ there is a representation of the trajectory that would be close enough to the original MTSP solution such that it is able to satisfy the constraint in \eqref{eq:optimization_initial_solution}. Therefore, we can always find feasible solutions to \eqref{eq:optimization_initial_solution} if we have a MTSP solution.

\subsection{2D Simulation Results}
\label{sec:results}
In this section, we demonstrate the results of the algorithm in two simulated 2D scenarios, one with one agent and three targets and the other one with three agents and 15 targets. All the internal states of the targets have the same state dynamics, evolving according to \eqref{eq:dynamics_phi} with
\begin{align*}
    A_i = \begin{bmatrix} -1 & -0.1 \\ -0.1 & 0.01 \end{bmatrix}, \quad Q_i = \textrm{diag}(1,1),
\end{align*}
and the agent observation models are given by \eqref{eq:observation_model_ij} with
\begin{align*}
    H_i = R_i=\textrm{diag}(1,1), \quad r_j=0.5,\quad \eta=10^{-3}.
\end{align*}

For each of the agents, their trajectories had the first five harmonics in each axis, i.e., $f_k=k$, $k=1,...,5$, $\forall\ j$. In the initial step of the optimization, the period $T$ was set to $1$. The initial coefficients $a_{j,k}^{e_p},b_{j,k}^{e_p}$ were obtained by solving the optimization problem in \eqref{eq:optimization_initial_solution}. The MTSP solution was obtained after $3000$ iterations of the genetic algorithm proposed in \cite{TANG2000267} for solving the associated MTSP. The initial position of each agent was set to coincide with the position of the first target in the solution of the MTSP. A constant descent stepsize $\kappa_l =10^{-4}$ was used in the gradient descent.

In the first scenario (with one agent and three targets), targets were located at positions $x_1=(0,0.5)$, $x_2=(0.5,0)$ and $x_3=(-0.5,0)$. Figures \ref{fig:results_3_target}-\ref{fig:cost_1agent} show the results we obtained. Figure \ref{fig:traj_3_targets} highlights how the trajectory changed from the initial one (an ellipse) to one with an almost triangular shape. Note, however, that not only the geometry of the trajectory is being optimized, but also the speed of the agent along the trajectory. From Fig. \ref{fig:position_3_target} we can see that the agent moves with higher speed when it is not visiting any target and at reduced speed (and the speed even completely vanishes) when it is close to the targets. Also, we can note that the trajectory in the last step of the optimization had a period lower than 1, which was the period on the initial optimization step. The mean estimation error over time for each of the targets is displayed in Fig. \ref{fig:covariance_3_target} and the cost along the optimization process is shown in Fig. \ref{fig:cost_1agent}.

\begin{figure*}[htp!]
    \centering
    \begin{subfigure}{0.32\textwidth}
        \centering\includegraphics[width=\textwidth]{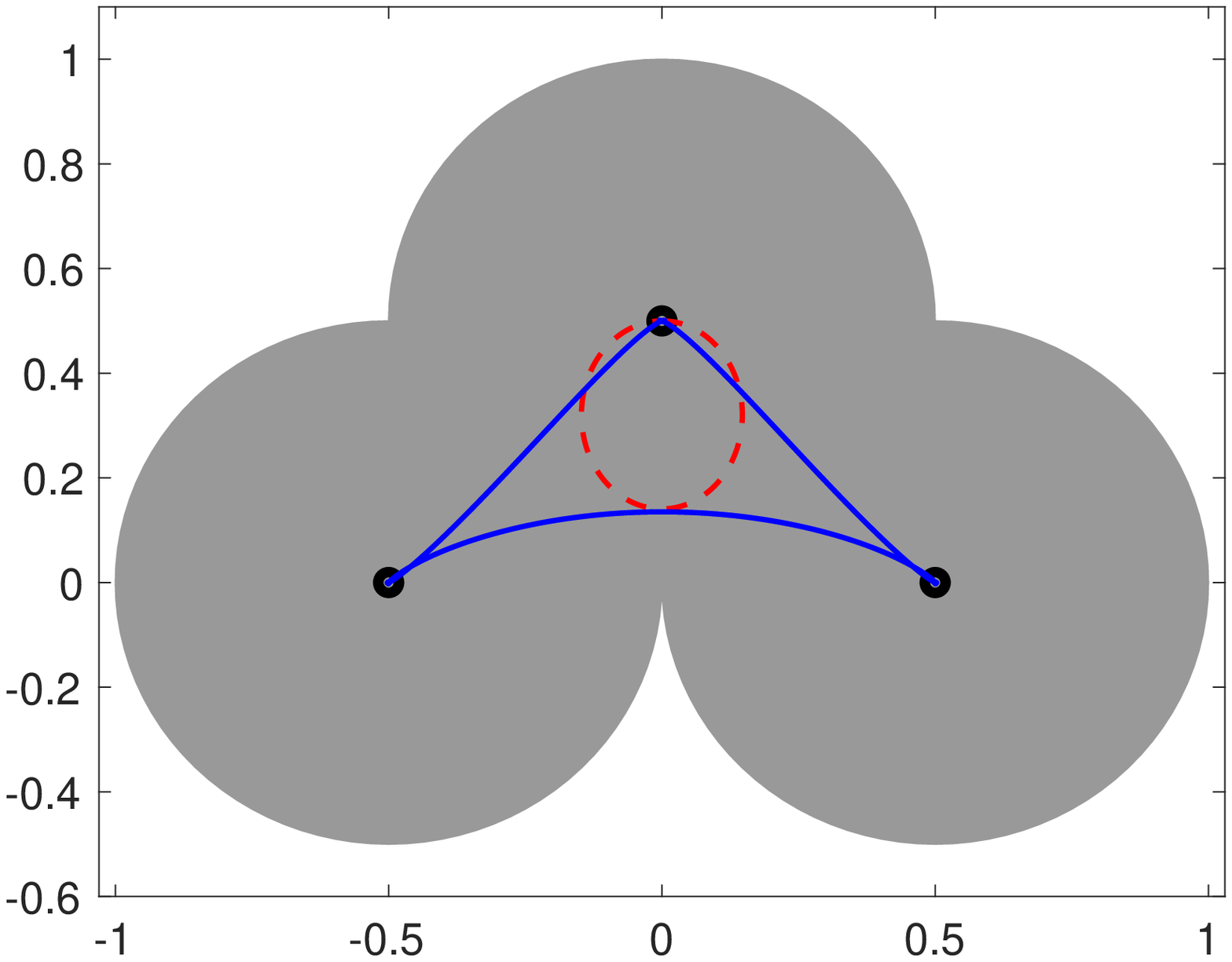}
        \caption{2D trajectory of the agent.}
        \label{fig:traj_3_targets}
    \end{subfigure}
    \begin{subfigure}{0.32\textwidth}
        \centering\includegraphics[width=\textwidth]{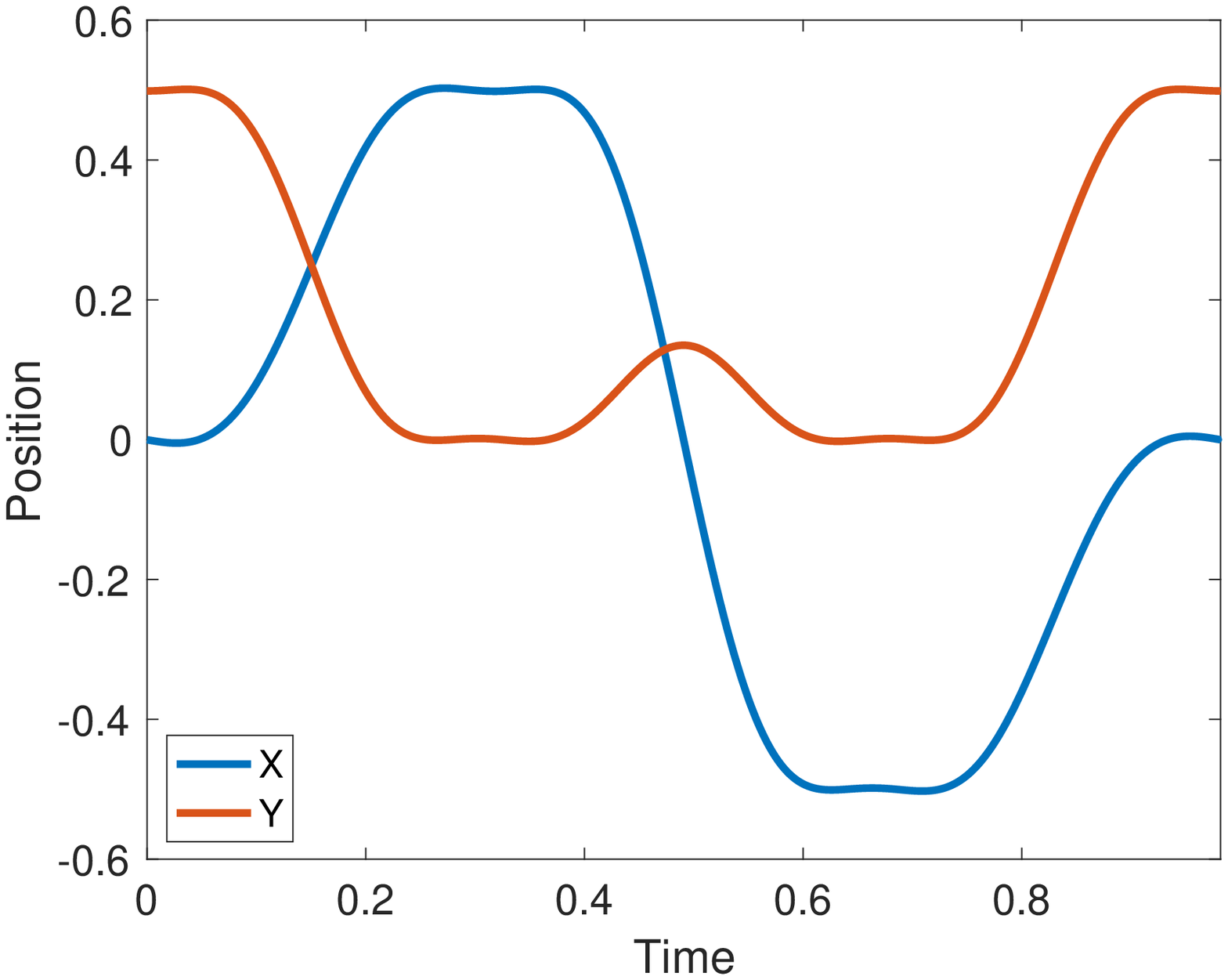}
        \caption{Agent positions as a function of time.}
        \label{fig:position_3_target}
    \end{subfigure}
    \begin{subfigure}{0.32\textwidth}
        \centering\includegraphics[width=\textwidth]{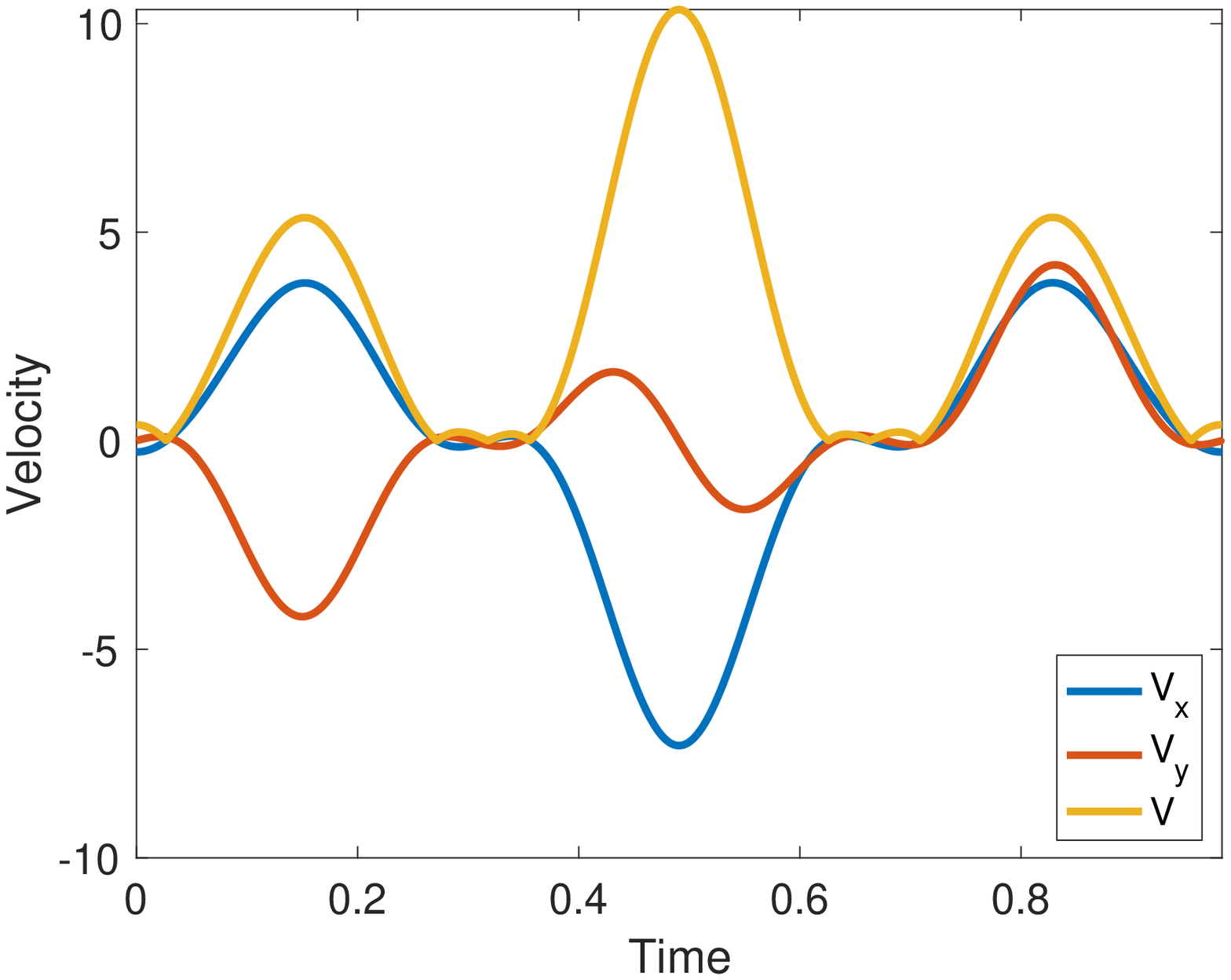}
        \caption{Agent velocity.}
        \label{fig:velocity_3_target}
    \end{subfigure}
    \caption{Simulation results with one agent and three targets. (a) Comparison of the 2-D trajectory of the agents in the initial (red) and final (blue) trajectories. The targets are marked in black and the gray area is the region where an agent can sense that target. (b) Agent trajectories at the final iteration in the $x$ (blue) and $y$ (red) directions. (c) Agent velocities at the final iteration in the $x$ (blue) and $y$ (red) directions and the resulting agent speed (yellow).}
    \label{fig:results_3_target}
\end{figure*}

\begin{figure}[htp!]
    \centering
    \includegraphics[width=0.32\textwidth]{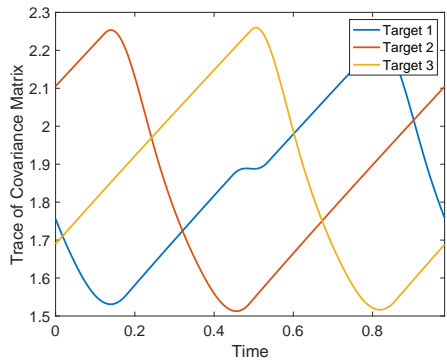}    
    \caption{Trace of the covariance for each target at the final step of the optimization in the scenario with one agent and three targets.}
    \label{fig:covariance_3_target}
\end{figure}

\begin{figure}[htp!]
    \centering
    \includegraphics[width=0.32\textwidth]{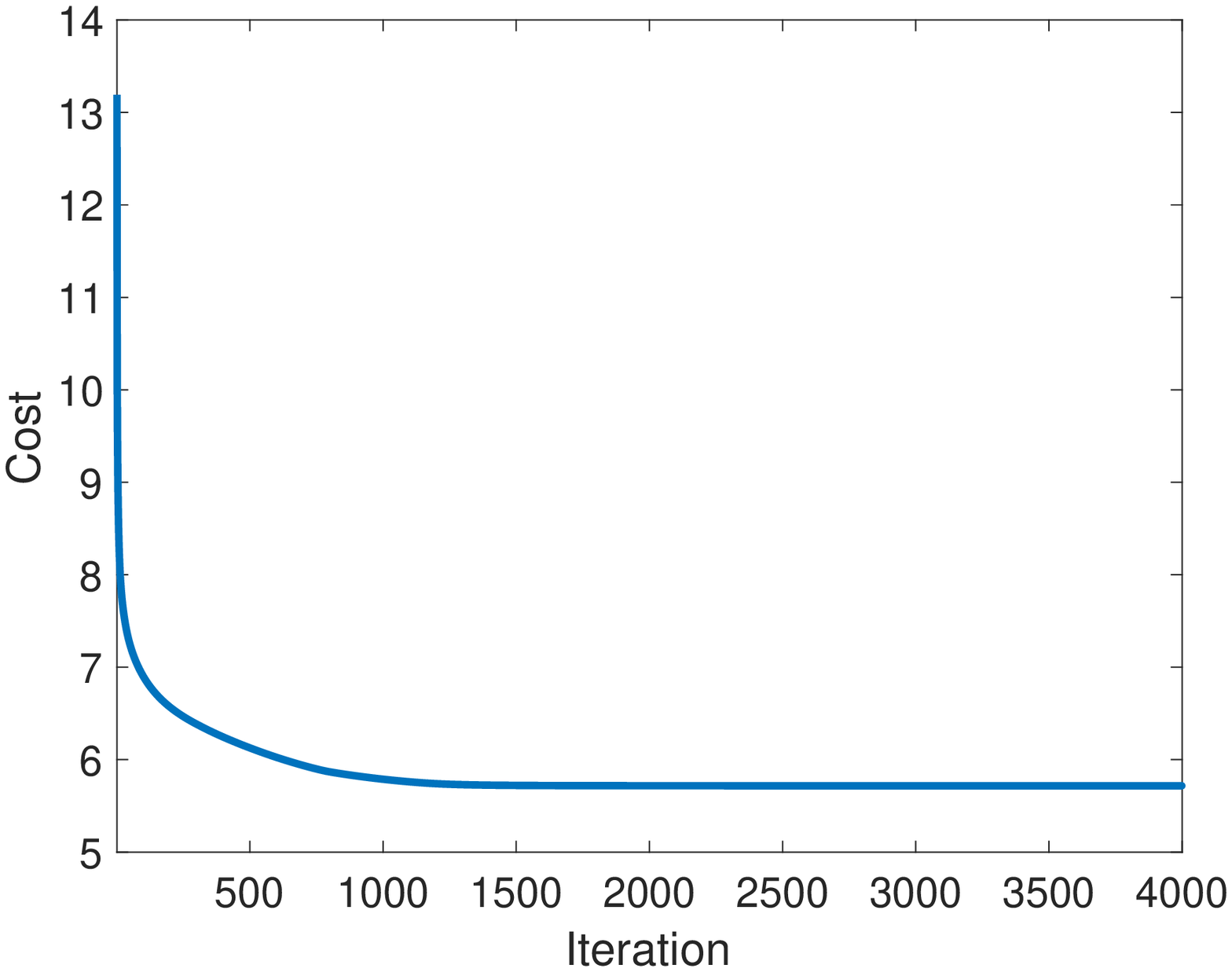}    
    \caption{Evolution of the cost function in the gradient descent optimization in the scenario with 3 targets and 1 agent.}
    \label{fig:cost_1agent}
\end{figure}

\begin{figure}[htp!]
    \centering
    \includegraphics[width=0.32\textwidth]{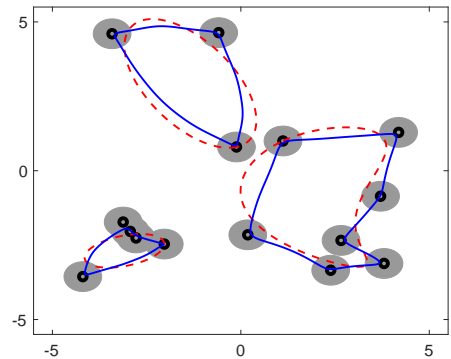}    
    \caption{Trajectories of the targets in the first (red dashed line) and last (blue continuous line) iterations of the gradient descent optimization on the scenario with 15 targets and 3 agents. The target's locations are marked in black and the grey shaded are represent the regions where the target can be sensed by an agent.}
    \label{fig:traj_15agents}
\end{figure}

\begin{figure}[htp!]
    \centering
    \includegraphics[width=0.32\textwidth]{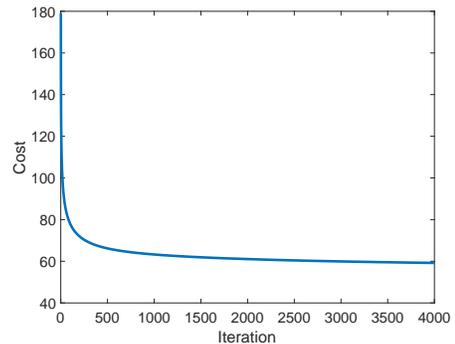}    
    \caption{Evolution of the cost function in the gradient descent optimization in the scenario with 15 targets and 3 agents.}
    \label{fig:cost_15agents}
\end{figure}

In the second scenario, the positions of the targets were generated randomly from independent uniform distributions ranging from $-5$ to $5$ in both axes. Fig. \ref{fig:traj_15agents} compares the trajectories of the agents in the first and last step of the gradient descent optimization, while Fig. \ref{fig:cost_15agents} shows the evolution of the cost as a function of the gradient descent step. The results of the optimization show that the solution of  \eqref{eq:optimization_initial_solution} led to smoother trajectories that still visited all the targets. The gradient descent changed the geometry of the trajectories but did not change the visiting order. As can be observed in Fig. \ref{fig:cost_15agents}, the cost has an abrupt reduction in the beginning of the optimization and then the convergence speed reduces significantly. The optimization process leads to very significant reductions of the cost, reducing it to less than one third of its initial value.

\subsection{3D Simulations Results}

In order to illustrate the extension of techniques proposed in this paper to higher dimensions, we present a result in a 3D environment, with 2 agents and 10 targets. The $A_i,\ Q_i,\ H_i,\ R_i$ matrices and $r_j$ are the same as in the 2D simulations. A constant gradient descent stepsize $\kappa_l=10^{-2}$ was used. The target locations were drawn from a uniform distribution in the cube with coordinates ranging from $[-5,5]$ in each axis. The trajectories after 4000 gradient descent iterations are shown in Fig. \ref{fig:results_3d} and the evolution of the cost is diplayed in Fig. \ref{fig:cost_3D}.

\begin{figure}[htp!]
    \centering
    \includegraphics[width=0.99\columnwidth]{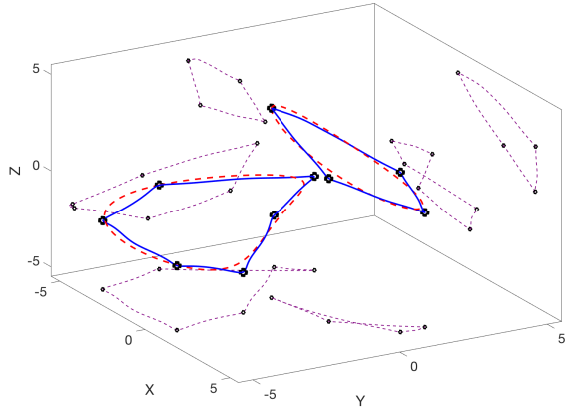}
    \caption{Simulation results in a 3D environment with two targets and ten agents. In red, the initial trajectory in the gradient descent optimization, in blue, the trajectory at the end of the optimization. The projection of the final agent trajectories in three planes is plotted in dashed purple.}
    \label{fig:results_3d}
\end{figure}

\begin{figure}[htp!]
    \centering
    \includegraphics[width=0.32\textwidth]{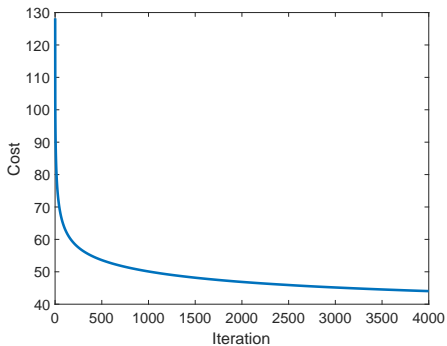}    
    \caption{Evolution of the cost function in the gradient descent optimization in the 3D scenario.}
    \label{fig:cost_3D}
\end{figure}

The 3D results follow a very similar trend of the 2D ones. The trajectories provided by the initialization procedure tend to be smoother, while the shape of the optimized ones are stiffer.

\section{Conclusion and Future Work}
\label{sec:conclusion}
We have addressed the problem of persistent monitoring from an infinite horizon perspective. We used a model that captures internal states of the targets evolving with linear stochastic dynamics and an observation model where the observation quality varies with distance. We derived necessary conditions for the convergence of the covariance matrix to a limit cycle as time goes to infinity. We also provided an algorithm for computing the cost gradient with respect to the parameters that define the trajectory. For a 1D environment, we showed that under some assumptions it is possible to fully characterize an optimal control by a finite set of parameters and used this as a basis for constructing an efficient parameterization. In higher dimensions, we proposed the use of Fourier curves for representing the trajectory. Our simulations illustrated the application of the proposed techniques in 1D, 2D and 3D scenarios, considering finite and infinite horizons for the cost.

Some challenges still remain for the framework presented in this paper. In ongoing research, we are studying how to efficiently select the gradient descent stepsize and also the feasibility and efficiency of local optimization methods other than gradient descent. We intend to study initialization methods that directly use uncertainty (instead of distance) as a criterion for generating initial schedules. We also plan to extend this paradigm to discrete time formulations and to investigate the feasiblity of distributed solutions. Lastly, we plan to study this problem when targets can also move, including movement models that are not fully deterministic.

\bibliographystyle{IEEEtran}
\bibliography{references.bib}

\appendices
\section{Proof of Optimality of Kalman Bucy Filter on the Persistent Monitoring Problem with Uncertain States}
\label{ap:proof_kbfilter}
\label{sec:optimality_kf_proof}
The set of all unbiased estimators $\hat{\phi}_i(t)$ of $\phi_i(t)$, as discussed in Sec. IV of \cite{athans1967direct}, is:
\begin{equation}
    \label{eq:app_dynamics_mean}
    \dot{\hat{\phi}}_i(t) = \left(A_i-G_i(t)\tilde{H}_i(t)\right)\hat{\phi}_i(t)+G_i(t)\tilde{z}_i(t),
\end{equation}
with $E[\hat{\phi}_i](0)=E[{\phi}_i(0)]$ and $G(t)$ a gain function that should be considered an input for the sake of optimality analysis. If $\Omega_i(t)=E[e_i(t)e_i'(t)]$, where $e_i=\hat{\phi}_i(t)-{\phi}_i(t)$, then
\begin{equation}
\begin{aligned}
    \label{eq:app_dynamics_covariance}
    \dot{\Omega}_i(t)=& \left(A_i-G_i(t)\tilde{H}_i(t)\right){\Omega}_i(t)+G_i(t)\tilde{R}_iG_i'(t)\\&+Q_i+{\Omega}_i(t)\left(A_i'-\tilde{H}_i(t)'G_i'(t)\right)
\end{aligned}
\end{equation}
and $\Omega_i(0)=\Omega_{i,0}$. Defining the following cost:
\begin{equation}
    J=\int_0^{t_f} \left(\sum_{i=1}^M\tr\left(\Omega_i(t')\right)+\beta u'(t')u(t') \right) \dt'
\end{equation}
The Hamiltonian is then
\begin{multline}
    \mathcal{H} = \sum_{i=1}^M\tr\left(\Omega_i(t)\right) + \beta u'(t)u(t) \\ +\sum_{i=1}^M\tr\left({\Gamma_i(t)\dot{\Omega}_i(t)}\right)+\sum_{j=1}^N\alpha_j(r)s_j(t),
\end{multline}
{where $\Gamma_i$ is the costate of $\Omega_i$. Using Pontryagin's minimum principle, 
at an optimal trajectory, since $G_i$ is unconstrained, we have
\begin{equation}
\label{eq:app_pmp_gain}
\frac{\partial \mathcal{H}^\star}{\partial G_i}=0.    
\end{equation}
Substituting the dynamics of the covariance matrix \eqref{eq:app_dynamics_covariance} on \eqref{eq:app_pmp_gain}, we get
\begin{equation}
    \label{eq:app_algebraic_hamilton_eq}
    -\Gamma_i\Omega_i\tilde{H}_i'-\Gamma_i'\Omega_i\tilde{H}_i'+\Gamma_i'G_i\tilde{R}_i+\Gamma_iG_i\tilde{R}_i=0.
\end{equation}
Now, again from the minimum principle,
\begin{multline}
        \dot{\Gamma}_i = -\frac{\partial \mathcal{H}}{\partial \Omega_i}-(A_i-G_i\tilde{H}_i)'\Gamma_i-\Gamma_i(A_i-G_i\tilde{H}_i)-I.
\end{multline}
Since $\Gamma_i(t_f)=0$ due to the boundary conditions of Pontryagin's minimum principle, the symmetric nature of this ODE allow us to see that $\Gamma_i$ will be symmetric for $t\in[0,t_f]$. Moreover, note that the ODE is linear and the single non-homogeneous term is -I. Since $\Gamma_i(t_f)=0$,
\begin{equation}
\begin{aligned}
    \Gamma_i(t)=&-\int_{t_f}^{t} \Phi'(t,t_f)\Phi(t,t_f) dt, \\ \Phi(a,b) =& \exp\left(\int_a^b {(A_i-G(\beta)\tilde{H}_i(\beta))}d\beta\right).
\end{aligned}
\end{equation}
This implies that $\Gamma_i(t)\succ0$ for $t\in[0,t_f)$.} Therefore, since $\Gamma_i(t)$ is invertible and symmetric, Eq. \eqref{eq:app_algebraic_hamilton_eq} can be reduced to
\begin{equation}
    \Omega_i\tilde{H}_i'+\Omega_i'\tilde{H}_i'=2G_i\tilde{R}_i.
\end{equation}
Since the covariance matrix $\Omega_i$ is also symmetric,
\begin{equation}
    G_i(t) = \Omega_i(t)\tilde{H}_i(t)\tilde{R}^{-1}_i(t)
\end{equation}
Plugging in this expression on \eqref{eq:app_dynamics_covariance} and \eqref{eq:app_dynamics_mean}, we get the usual Kalman-Bucy filter equations, which along with the initial conditions $\Omega_i(0)=\Omega_{i,0}$ and $\hat{\phi}_i(0)=E[{\phi}_i(0)]$, have unique solutions.
    
\section{Existence of Steady State Covariance Derivatives}
\label{ap:existence_derivatives}
{In this appendix, we discuss the existence of the gradients of the steady state covariance matrix. Note that, if in a periodic trajectory $\eta_i(q)=0\ \forall q\in[0,1]$ {(i.e., target $i$ is never visited)}, the existence of the steady state covariance matrix is not guaranteed by Prop. \ref{prop:unique_attractive_sol_riccati_eq}. Obviously, if the steady state covariance does not exist, its derivative will also not exist. This illustrates the fact that the existence $\frac{\partial \bar{\Omega}_i}{\partial \theta}$ is not guaranteed.  What we show in this appendix is that, under very natural assumptions, the derivative $\frac{\partial \bar{\Omega}_i}{\partial \theta}$ exists for the parameters that belong to the interior of the set of parameters that will lead to convergence of the steady state covariance, except for a set of zero measure.}

{
Since here we analyze the behavior of the steady state covariance with respect to parameter variations, we will use a notation that explicitly shows the dependence of the variables with the parameters. For example, $\bar{\Omega}_i$ is a function of $q$ and of the parameters $\Theta$ and, hence, it will be denoted as $\bar{\Omega}_i(q;\Theta)$.}


{
We define the set of parameters for which the steady state covariance is guaranteed to exist as:
\begin{multline}
    \vartheta = \{\Theta\ |\ \eta_i(q,\tilde{\Theta})>0\\\text{ for some non-degenerate interval $q\in[a,b]$}\},
\end{multline}
and $\Psi$ as the interior of the set $\vartheta$.}

{Our goal is to show that, for any $\Theta\in\Psi$, the partial derivatives $\frac{\partial\bar{\Omega}_i(q;\Theta)}{\partial \theta_d}$ exist locally. From Prop. \ref{prop:solution_lyapunov_eq}, we know that, when this partial derivative exists, it is equal to $\Sigma(q;\Theta)$. We also know that $\Sigma(q;\Theta)$ is well defined for any $\theta\in\Psi$. We now make the following assumption about the regularity of $\Sigma$:}
\begin{assumption}
    $\Sigma(q;\Theta)$ is locally Riemann integrable for $\Theta\in\Psi$.
\end{assumption}
{
In light of Proposition \ref{prop:solution_lyapunov_eq}, Assumption 3 means that the parameterizations that we consider do not allow for an infinite number of discontinuities of $\Sigma_h(q;\Theta)$ and $\Sigma_{ZI}(q;\theta)$. Note that, due to the linear nature of their underlying differential equations, $\Sigma_h(q;\Theta)$ and $\Sigma_{ZI}(q;\theta)$ are bounded for any $\Theta\in\Psi.$ Therefore, $\Sigma(q;\Theta)$ is also bounded.}

\begin{proposition}
    {Under Assumptions 1, 2 and 3, the partial derivative $\frac{\partial\bar{\Omega}_i(q;\Theta)}{\partial \theta_d}$, $q\in[0,1]$ and $\Theta\in\Psi$, exists almost everywhere in $[0,1]\times\Psi$.}
\end{proposition}
\begin{proof}
{By construction, we pick two parameter sets $\Theta_1$ and $\Theta_2$, such that any convex combination of $\Theta_1$ and $\Theta_2$ belongs to $\Psi$. Additionally, since our goal is to compute the partial derivative with respect to $\theta_d$, we pick $\Theta_2$ such that it differs from $\Theta_1$ only in its $d$-th coordinate. Since the set $\Psi$ is open, if we pick any $\Theta_1\in\Psi$, we can always find a $\Theta_2$ that fullfills the aforementioned properties.}

{We define the function $\Upsilon(q;\Theta_2)$ (which later we will show $\Upsilon(q;\Theta_2)=\bar{\Omega}_i(q;\Theta_2)$) as:
\begin{equation}
    \label{eq:fundamental_calculus_integral}
    \Upsilon(q;\Theta_2)=\bar{\Omega}_i(q;\Theta_1)+\int_0^1\Sigma(q;\Theta_1+\xi(\Theta_2-\Theta_1))d\xi.
\end{equation}
Note that, if $\Upsilon(q;\Theta_2)=\bar{\Omega}_i(q;\Theta_2)$ for generic $\Theta_1,\Theta_2$, then $\Sigma(q;\Theta)=\frac{\partial\bar{\Omega}_i(q;\Theta)}{\partial \theta_d}$ almost everywhere, since $\Sigma(q;\Theta)$ plays the role of a partial derivative in Eq. \eqref{eq:fundamental_calculus_integral}.}

{
$\bar{\Omega}_i(q;\Theta_2)$ is uniquely defined by satisfying the differential equation \eqref{eq:scaled_riccati_diff_eq} and being periodic with period one. We then show that $\Upsilon(q,\Theta_2)$ also satisfies both of these properties, which imply that indeed $\Upsilon(q,\Theta_2)=\bar{\Omega}_i(q;\Theta_2)$.}

{
First, notice that $\Upsilon(0;\Theta_2)=\Upsilon(1;\Theta_2)$ since $\bar{\Omega}_i(0;\Theta_1)=\bar{\Omega}_i(1;\Theta_1)$ and $\Sigma(0,\Theta)=\Sigma(1,\Theta)$, for any $\Theta\in\Psi$. Also, since $\Sigma(q;\Theta)$ is a solution of $\eqref{eq:derivative_omega_i}$,
\begin{multline}
    \label{eq:integral_sigma_dot_along_path}
    \int_0^1\dot{\Sigma}(q;\Theta_1+\xi(\Theta_2-\Theta_1))d\xi = \dot{\bar{\Omega}}_i(q,\Theta_2)-\dot{\bar{\Omega}}_i(q,\Theta_1).
\end{multline}
Therefore, taking the derivative of \eqref{eq:fundamental_calculus_integral} with respect to $q$ and substituting \eqref{eq:integral_sigma_dot_along_path}, we get
\begin{equation}
    \dot\Upsilon(q,\Theta_2) = \dot{\bar{\Omega}}_i(q,\Theta_2).
\end{equation}
Hence we conclude that $\Upsilon(q,\Theta_2)=\bar{\Omega}_i(q;\Theta_2)$, and, as a consequence, $\frac{\partial\bar{\Omega}_i(q;\Theta)}{\partial \theta_d}$ exists almost everywhere in $\Psi$. Additionally, as already stated in Prop. \ref{prop:solution_lyapunov_eq}, $\frac{\partial\bar{\Omega}_i(q;\Theta)}{\partial \theta_d}=\Sigma(q,\Theta)$ wherever it exists.}
\end{proof}

\end{document}